%% file: det-jsc-new.tex
\documentclass[letterpaper,10pt]{article}
\usepackage[margin=1in]{geometry}
\usepackage{bold-extra}

\usepackage{latexsym,amsfonts,amssymb,amsmath,amscd,epic,graphics}
\usepackage{psfrag,epsfig}
\usepackage{graphicx}
\usepackage{amsthm,amsfonts,color}  

\usepackage{xspace}
\usepackage[naturalnames]{hyperref} 
\usepackage[sort,numbers]{natbib} 
\usepackage{algorithmic} 
\usepackage[boxruled,boxed,algo2e]{algorithm2e}
\usepackage{ifthen}
\usepackage{lscape}
\usepackage{marvosym}
\usepackage{multirow}

\usepackage{subfig}

\include{notations-new}
\newtheorem{lemma}{Lemma}
\newtheorem{theorem}[lemma]{Theorem}

\newtheorem{proposition}[lemma]{Proposition}
\newtheorem{corollary}[lemma]{Corollary}

\begin{document}

  \title{On the asymptotic and practical complexity of solving bivariate systems
  over the reals}

  \author{Dimitrios I.~Diochnos $^1$ \and Ioannis Z.~Emiris $^2$ \and Elias P.~Tsigaridas $^3$}

  \date{$^1$ {University of Illinois at Chicago, USA}\\ \texttt{diochnos (AT) math.uic.edu}\\[0.5em] 
$^2$ {National Kapodistrian University of Athens, HELLAS}\\ \texttt{emiris (AT) di.uoa.gr}\\[0.5em]
$^3$ INRIA Sophia-Antipolis, FRANCE\\ \texttt{elias.tsigaridas (AT) sophia.inria.fr}\\[0.5em]
\vspace{0.7cm}
June 12, 2008}

\maketitle

\thispagestyle{empty}

  \begin{abstract}
    This paper is concerned with exact real solving of well-constrained, bivariate
    polynomial systems.  The main problem is to isolate all common real roots in rational
    rectangles, and to determine their intersection multiplicities.  We present three
    algorithms and analyze their asymptotic bit complexity, obtaining a bound of
    $\sOB(N^{14})$ for the purely projection-based method, and $\sOB(N^{12})$ for two
    sub\-result\-ant-based methods: this 
    notation
    ignores polylogarithmic factors, 
    where $N$ bounds the degree and the bitsize of the polynomials.  
    The previous record bound was $\sOB(N^{14})$.

    Our main tool is signed subresultant sequences.
    We exploit recent advances on the complexity of
    univariate root isolation, and extend them to %
    sign evaluation of bivariate polynomials over two algebraic numbers,
    and real root counting for polynomials over an extension field.  Our algorithms apply to
    the problem of simultaneous inequalities; they also compute the topology of real plane
    algebraic curves in $\sOB( N^{12})$, whereas the previous bound was $\sOB( N^{14})$.

    All algorithms have been implemented in \maple, in conjunction with numeric filtering.
    We compare them against \gbrs, system solvers from \synaps, and \maple
    libraries \func{insulate} and \func{top}, which compute curve topology.  Our software is
    among the most robust, and its runtimes are comparable, or within a small constant
    factor, with respect to the C/C++ libraries.

    \bigskip

  \end{abstract}

\newpage

\section{Introduction} \label{sec:introduction}

The problem of well-constrained polynomial system solving is fundamental.
However, most of the algorithms treat the general case or
consider solutions over an algebraically closed field.
We focus on real solving of bivariate polynomials in order to provide precise
complexity bounds and study different algorithms in practice.
We expect to obtain faster algorithms than in the general case.
This is important in several applications ranging from
nonlinear computational geometry to real quantifier elimination.
We suppose relatively prime polynomials for simplicity,
but this hypothesis is not restrictive.
A question of independent interest,
which we tackle,
is to compute the topology of a
real plane algebraic curve. %

Our algorithms isolate all common real roots inside non-overlapping rational rectangles,
output them as pairs of algebraic numbers,
and determine the intersection multiplicity per root.
Algebraic numbers are represented by an isolating
interval and a square-free polynomial.

In this paper, $\OB$ means bit complexity and $\sOB$ means that we
are ignoring polylogarithmic factors.
We derive a bound of $\sOB(N^{12}),$
whereas the previous record bound was $\sOB(N^{14})$ \cite{VegKah:curve2d:96}, see
also \cite{BPR06},
derived from the closely related problem of computing the topology of
algebraic curves, 
where $N$ bounds the degree and the bitsize of the input polynomials.
This approach depends on Thom's encoding.
We choose the isolating interval representation, since it is more intuitive, 
and is used in applications.
In \cite{VegKah:curve2d:96}, it is stated that
``isolating intervals provide worst [sic] bounds''.
It is widely believed that isolating intervals 
do not produce good theoretical results.
Our work suggests that isolating intervals should be re-evaluated.

Our main tool is signed subresultant sequences (closely related to
Sturm-Habicht sequences), extended
to several variables by binary segmentation.
We exploit the recent advances on univariate root isolation,
which reduced complexity by one to three orders
of magnitude to $\sOB(N^6)$
\cite{Yap:SturmBound:05,ESY:descartes,emt-lncs-2006}.
This brought complexity closer to $\sOB(N^4)$, which is achieved
by numerical methods \cite{Pan02jsc}.

In \cite{KoSakPat:system2:05},
$2\times 2$ systems are solved and the multiplicities computed
under the assumption that a generic shear has been obtained,
based on \cite{sakkalis90:alg-curves}.
In \cite{WolPhd}, $2\times 2$ systems of bounded degree were studied, 
obtained as projections of the arrangement of 3D quadrics.
This algorithm is a precursor of ours, see also \cite{et-casc-2005},
except that matching and multiplicity computation was simpler.
In \cite{MouPav:TR-sbd:05}, a subdivision algorithm is proposed,
exploiting the properties of the Bernstein basis, with unknown bit complexity,
and arithmetic complexity based on the characteristics of the graphs of the polynomials.
For other approaches based on multivariate Sturm sequences the reader
may refer to e.g. \cite{Miln92,PeRoSz93}. 

For
determining the topology of a real algebraic plane curve, %
the best bound is $\sOB(N^{14})$ \cite{BPR06,VegKah:curve2d:96}.
In \cite{WolSei:topology:05} three projections are used;
this is implemented in \func{insulate}, with which we make several comparisons.
Work in \cite{EKW:curves} offers an efficient
implementation of resultant-based methods,
whereas
Gr{\"o}bner bases %
are employed in 
\cite{fc-jcf-mp-fr-isrsps-06}.
To the best of our knowledge, 
the only result %
for
topology determination using isolating intervals 
is \cite{ArnMcC:topology:88}, where a $\sOB( N^{30})$ bound is proved.

We establish a bound of $\sOB( N^{12})$ using the isolating interval representation.
It seems that the complexity in \cite{VegKah:curve2d:96}
could be improved to $\sOB( N^{10})$ using fast multiplication algorithms, 
fast algorithms for computations of signed subresultant sequences
and improved bounds for the bitsize of the integers appearing in computations. 
To put our bounds into perspective, 
the input 
size
is in $\OB( N^3)$,
and the total bitsize of all output isolation points for univariate solving
is 
in 
$\sOB( N^2 )$, and this is tight.
Notice that lower bounds in real algebraic geometry refer almost exclusively to arithmetic complexity \cite{1997-buergisser}.
 
The main contributions of this paper are the following:
Using the aggregate separation bound, we improve the complexity for 
computing the sign of a polynomial
evaluated over all real roots of another (Lemma \ref{lem:sign-at-1-all}).
We establish a complexity bound for bivariate sign evaluation (Theorem \ref{th:biv_sign_at}),
which helps us derive bounds for
root counting in an extension field (Lemma \ref{lem:count-alpha}) 
and for
the problem of simultaneous inequalities (Corollary \ref{cor:biv-inequalities}).
We study the complexity of bivariate polynomial real solving,
using three projection-based algorithms: a straight\-forward grid method
(Theorem \ref{th:grid_solve}), a specialized RUR 
(Rational Univariate Representation)
approach (Theorem \ref{th:mrur-solve}),
and an improvement of the latter using fast GCD (Theorem \ref{th:grur}).
Our best bound is $\sOB( N^{12})$;
within this bound, we also compute the root multiplicities.
Computing the topology of a real plane algebraic curve is in $\sOB(
N^{12})$ (Theorem \ref{th:topology}). 

We implemented in \maple a package for computations with real algebraic
numbers and for implementing our algorithms.
It is easy to use and integrates seminumerical filtering
to speed up computation when the roots are well-separated.
It guarantees exactness and completeness of results;
moreover, the runtimes %
are quite
encouraging.
We illustrate it by experiments against well-established \cc/\cpp libraries
\gbrs and \func{synaps}.
We also examine \maple libraries \func{insulate}
and \func{top}, which compute curve topology.
Our software is %
among the most
robust;
its runtime is within a small constant factor
with respect to the fastest \cc/\cpp library.

The next section presents basic results concerning real solving
and operations on univariate polynomials.
We extend the discussion to several variables,
and focus on bivariate polynomials.
The algorithms for bivariate solving and their analyses appear in
Section \ref{sec:biv-solving}, followed by applications to real-root
counting, simultaneous inequalities and the topology of curves.
Our implementation and experiments appear in Section \ref{sec:implementation}.

A preliminary version of our results appeared in~\cite{det-issac-2007}.

\section{Univariate polynomials} \label{sec:preliminaries}
For $f \in$ $\ZZ[y_1, \dots$, $y_k, x]$, $\dg( f)$ denotes its total degree,
while $\dg_{x}(f)$ %
denotes its degree w.r.t.~$x$.
\bitsize{f} bounds the bitsize of the
coefficients of $f$ (including a bit for the sign).
We assume $\lg{ ( \dg( f))} = \OO( \bitsize{ f})$.
For $\rat{a} \in \QQ$, $\bitsize{ \rat{a}}$ is
the maximum bitsize of numerator and denominator.
Let \Multiply{\tau} denote the bit complexity of multiplying two integers of size
$\tau$, and \Multiply{d, \tau} the complexity of multiplying
two univariate  polynomials of degrees $\le d$
and coefficient bitsize $\le \tau$.
Using \textsc{FFT}, $\Multiply{\tau} = \sOB( \tau )$
and
$\Multiply{d, \tau} = \sOB( d \tau)$.

Let $f, g \in \ZZ[x]$,
$\dg( f) = p \geq q = \dg( g)$ and $\bitsize{f}, \bitsize{g} \le \tau$.
We use \rem{f, g} and \quo{f, g} for the Euclidean remainder and quotient, respectively.
The {\em signed polynomial remainder sequence} of $f, g$ is
$R_0 = f$, $R_1 = g$, $R_2 = -\rem{f, g}$, $\dots$, $R_k = -\rem{R_{k-2}, R_{k-1}}$, 
where $\rem{R_{k-1}, R_{k}} =0$.
The {\em quotient sequence} contains 
$Q_i = \quo{R_i, R_{i+1}}$, $i=0\ldots k-1$, and the {\em quotient boot} is $(Q_0, \dots, Q_{k-1}, R_k)$.

We consider signed subresultant sequences \cite{BPR06}, 
which contain polynomials similar to the polynomials 
in the signed polynomial remainder sequence; see  
\cite{GathenLucking:Subresultants:03} for a unified approach to subresultants.
They achieve better bounds on the coefficient bitsize 
and have good specialization properties.
In our implementation we use Sturm-Habicht (or Sylvester-Habicht) sequences, 
see e.g.~\cite{VegLomRecRoy:StHa:89,BPR06,LickteigRoy:FastCauchy:01}.  
By $\SR(f, g)$ we denote the signed subresultant sequence,
by $\sr(f,g)$ the sequence of the principal subresultant coefficients,
by $\SRQ( f, g)$ the corresponding quotient boot.
By $\SR_j( f, g)$, or simply $\SR_j$ if the corresponding polynomials can be easily
deduced from the context we denote an element of the sequence; 
similarly for $\sr_j$ and $\SRQ_j$. 
Finally, by $\SR( f, g;\, \rat{a})$ we denote the evaluated sequence over $\rat{a}\in\QQ$.
If the polynomials are multivariate, then these %
sequences are considered w.r.t.~$x$, except if explicitly stated otherwise.

\begin{proposition} \label{pr:SR-computation}
  \cite{LickteigRoy:FastCauchy:01,Reischert:subresultant:97}
  Assuming $p \geq q$, $\SR(f, g)$ is computed in $\sOB(p^2 q \tau)$ and
  $\bitsize{\SR_j(f, g)} = \OO( p \tau)$.
  For any $f,g$, their quotient boot, any polynomial in $\SR( f, g)$,
  their resultant, and their $\gcd$ are computed in $\sOB(p q \tau)$.
\end{proposition}

The following proposition is a slightly modified version of the one that appeared in 
\cite{LickteigRoy:FastCauchy:01,Reischert:subresultant:97}.

\begin{proposition} \label{pr:SR-fast-evaluation}
  Let $p\ge q$. We can compute $\SR(f, g;\rat{a})$,
  where $\rat{a} \in \QQ \cup \{ \pm \infty \}$ and $\bitsize{\mathsf{a}} =  \sigma$, 
  in $\sOB(p q \tau + q^2 \sigma + p^2 \sigma )$.
  If $f( \rat{ a})$ is known, then the bound becomes $\sOB(pq\tau + q^2\sigma).$
\end{proposition}

\begin{proof} %
  Let $\SR_{q+1} = f$ and $\SR_{q} = g$. For the moment we forget $\SR_{q+1}$.
  We may assume that $\SR_{q-1}$ is computed, since the cost of computing 
  one element of $\SR(f, g)$ is the same as that of computing $\SRQ(f, g)$
  (Pr.~\ref{pr:SR-computation})
  and we consider the cost of evaluating the sequence $\SR(g, \SR_{q-1})$
  on \rat{a}.

  We follow \citet{LickteigRoy:FastCauchy:01}.
  For two polynomials $A, B$ of degree bounded by $D$ and bitsize bounded by $L$,
  we can compute $\SR( A, B ;\rat{a})$, where $\bitsize{ \rat{ a}} \leq L$, 
  in $\sOB( \Multiply{ D, L})$.
  In our case $D = \OO( q)$ and $L = \OO(  p \tau +  q \sigma)$,
  thus the total costs is
  $\sOB( p q \tau + q^2 \sigma)$.

  It remains to compute the evaluation $\SR_{q+1}( \rat{ a}) = f( \rat{ a})$. 
  This can be done using Horners' scheme in $\sOB( p \max\{ \tau, p \sigma\})$.
  Thus, the whole procedure has complexity 
  \begin{displaymath}
    \sOB( p q \tau + q^2 \sigma + p \max\{ \tau, p \sigma\}) ,
  \end{displaymath} 
  where the term $p\tau$ is dominated by $pq\tau$. 
\end{proof}

When $q > p$, $\SR( f, g)$ is $f, g, -f, -(g \bmod (-f)) \dots$,
thus $\SR( f, g;\, \rat{a})$ starts with a sign variation irrespective of
$\sign( g(\rat{a}))$.
If only the sign variations are needed, there is no need to evaluate $g$, so
Proposition \ref{pr:SR-fast-evaluation} yields $\sOB( p q \tau + p^2 \sigma )$.
Let $L$ denote a list of real numbers. 
$\var(L)$ denotes the number of (possibly modified, 
see e.g.~\cite{BPR06,VegLomRecRoy:StHa:89,VegKah:curve2d:96}) sign variations.

\begin{corollary} \label{cor:SR-fast}
  For any $f,g$, $\var( \SR( f, g; \rat{a}))$ is computed
  in $\sOB( p q \tau + \min\{ p, q\}^2  \sigma)$,
  provided $\sign( f( \rat{a}))$ is known.
\end{corollary}

We choose to represent a real algebraic number $\alpha\in\ALG$
by the {\em isolating interval} representation.
It includes a square-free polynomial which vanishes on $\alpha$ and a (rational) interval
containing $\alpha$ and no other root. By $f_{red}$ we denote the square-free part of $f$.

\begin{proposition} \label{pr:solve-1} \cite{Yap:SturmBound:05,ESY:descartes,emt-lncs-2006}
  Let $f\in\ZZ[x]$ have degree $p$ and bitsize $\tau_f$.
  We compute the isolating interval representation of its real roots
  and their multiplicities in $\sOB( p^6 + p^4 \tau_f^2)$.
  The endpoints of the isolating intervals have bitsize $\OO( p^2 + p \, \tau_f )$
  and $\bitsize{ f_{red}} = \OO( p + \tau_f )$.  
\end{proposition}

Notice that after real root isolation, the sign of the square-free part $f_{red}$ over the
interval's endpoints, say $[ \rat{ a}, \rat{ b}]$ is known; 
moreover, $f_{red}( \rat{ a}) f_{red}( \rat{b}) < 0$. 
The following proposition takes advantage of this fact and is a refined version of similar
proposition in e.g. \cite{BPR06,emt-lncs-2006}.

\begin{corollary} \label{cor:sign-at-1} 
  Given a real algebraic number $\alpha \cong (f, [\rat{a}, \rat{b}])$, 
  where $\bitsize{\rat{a}} = \bitsize{ \rat{b} } = \OO( p \tau_f)$,
  and $g \in \ZZ[x]$, such that $\dg( g) = q$ and $bitsize{ g} = \tau_g$,
  we compute $\sign( g (\alpha))$
  in bit complexity $\sOB( pq \max\{\tau_f, \tau_g\} + p \min\{p, q\}^2 \tau_f )$.
\end{corollary}

\begin{proof} %
  Assume that $\alpha$ is not a common root of $f$ and $g$ in $[\rat{a}, \rat{b}]$,
  then it is known that 
  \begin{displaymath}
    \sign( g( \alpha)) =  [ \var( \SR(f, g; \rat{a})) - \var( \SR(f, g; \rat{b})) ] \sign( f'( \alpha)).
  \end{displaymath}
  Actually the previous relation holds in a more general context, 
  when $f$ {\em dominates} $g$, see \cite{Yap2000} for details.
  Notice that $\sign( f'( \alpha)) = \sign( f( \rat{b})) - \sign( f(\rat{b}))$, which is known
  from the real root isolation process.
  
  The complexity of the operation is dominated by the
  computation of $\var( \SR( f, g; \rat{ a}))$ and $\var( \SR( f, g; \rat{ b}))$,
  i.e. we compute $\SRQ$ and evaluate it on $\rat{ a}$ and $\rat{ b}$.
  
  As explained above, there is no need to evaluate the polynomial of the largest degree, 
  i.e.\ the first (and the second if $p < q$) of $\SR( f, g)$ over $\rat{ a}$ and $\rat{b}$.  
  The complexity is that of Corollary \ref{cor:SR-fast}, i.e.\
  $ \sOB( p q \max\{ \tau_f, \tau_g\} + \min\{p, q \}^2 p\, \tau_f ). $
  Thus the operation costs two times the complexity of the evaluation
  of the sequence over the endpoints of the isolating interval.

  If $\alpha$ is a common root of $f$ and $g$, or if $f$ and $g$ are
  not relative prime, then their gcd, which is the last non-zero
  polynomial in $\SR(f, g)$ is not a constant. Hence, we evaluate $\SR$
  on $\rat{ a}$ and $\rat{ b}$, we check if the last polynomial is not
  a constant and if it changes sign on $\rat{ a}$ and $\rat{ b}$. 
  If this is the case, then $\sign( g( \alpha)) = 0$.
  Otherwise we proceed as above.
\end{proof}

Proposition \ref{pr:solve-1} expresses the state-of-the-art in univariate root isolation.
It relies on fast computation of polynomial sequences
and the Davenport-Mahler-Mignotte bound, see \cite{Dav:TR:85} for the first version of
this bound.
The following lemma, a direct consequence of Davenport-Mahler-Mignotte bound, is crucial.
\begin{lemma}[Aggregate separation] \label{lem:very-important} 
  Given $f\in\ZZ[x]$,
  the sum of the bitsize of {\em all} isolating points of the real roots of $f$
  is $\OO( p^2 + p\, \tau_f)$. %
\end{lemma} 

\begin{proof}
  Let there be $r\le p$ real roots.
  The isolating point, computed by a real root isolation subdivision algorithm
  \cite{Yap:SturmBound:05,ESY:descartes,emt-lncs-2006}, between two consecutive real
  roots, say $\alpha_j$ and  $\alpha_{j+1}$, 
  is of magnitude at most $\frac{1}{2}| \alpha_j - \alpha_{j+1}| := \frac{1}{2} \Delta_j$.
  Thus their product is $\frac{1}{2^r} \prod_{j=1}^{r-1} \Delta_j$.
  Using the Davenport-Mahler-Mignotte bound, the product is bounded from below, 
  that is  $\prod_j \Delta_j \geq 2^{-\OO( p^2 + p \tau_f)}$.
  Taking logarithms, we conclude the proof.
\end{proof}

We present a new complexity bound on evaluating the sign of a polynomial $g(x)$ over
a set of algebraic numbers, which have the same defining polynomial,
namely over all real roots of $f(x)$.  
It suffices to evaluate $\SR( f, g)$ over all the isolating endpoints of $f$.
The obvious technique, e.g.~\cite{emt-lncs-2006}, see also
\cite{Sakkalis:algnum:89,BPR06}, is to apply 
Corollary \ref{cor:sign-at-1} $r$ times, where $r$ is the number of real roots of $f$.
But we can do better by applying Lemma \ref{lem:very-important}:%

\begin{lemma} \label{lem:sign-at-1-all}
  Let $\tau=\max\{p, \tau_f,\tau_g\}$.
  Assume that we have isolated the $r$ real roots of $f$ and we know
  the signs of $f$ over the isolating endpoints.
  Then, we can compute the sign of $g$ over all $r$ roots of $f$ in
  $\sOB(p^2 q\tau )$.
\end{lemma}

\begin{proof}
  Let $s_j$ be the bitsize of the $j$-th endpoint, where $0 \leq j \leq r $.  
  The evaluation of $\SR(f,g)$ over this endpoint, by
  Corollary \ref{cor:SR-fast}, costs  $ \sOB( p q \tau + \min\{p,q\}^2 s_j)$. 
  To bound the overall cost, we sum over all isolating points.
  The first summand is $\sOB( p^2 q \tau)$.
  By Proposition \ref{lem:very-important}, the second summand
  becomes  $\sOB( \min\{p,q\}^2 (p^2 + p \tau_f))$ and is dominated.
\end{proof}

\section{Multivariate polynomials} \label{sec:multivariate}
In this section, we extend the results of the previous section to multivariate polynomials,
using binary segmentation \cite{Reischert:subresultant:97}.
Let $f, g \in (\ZZ[y_1, \dots, y_k])[x]$ with
$\dg_x( f) = p \geq q = \dg_x( g)$,
$\dg_{y_i}( f) \leq d_i$ and $\dg_{y_i}( g) \leq d_i$.
Let $d = \prod_{i=1}^{k}{d_i}$ and $\bitsize{f}, \bitsize{g} \leq \tau$.
The $y_i$-degree of every polynomial 
in $\SR( f, g)$ is bounded by $\dg_{y_i}( \res( f, g)) \leq  (p+q) d_i$.
Thus, the homomorphism $\psi: \ZZ[y_1, \dots, y_k] \rightarrow \ZZ[y]$, where
\begin{displaymath}
  y_1 \mapsto\ y, \, y_2 \mapsto\ y^{(p + q)d_1},\,
  \dots \,, y_k \mapsto\ y^{(p + q)^{k-1} d_1 \cdots d_{k-1}} ,
\end{displaymath}
allows us to decode $\res( \psi( f), \psi( g)) = \psi( \res( f, g))$
and obtain $\res( f, g)$.
The same holds for every polynomial in $\SR( f, g)$.
Notice that 
$\psi( f), \psi( g) \in (\ZZ[y])[x]$ have $y-$degree
less or equal to
$(p+q)^{k-1} d$ since,
in the worst case, $f$ or $g$ contains 
a
monomial of the form
$y_1^{d_1} \, y_2^{d_2} \dots y_k^{d_k}$.
Thus, $\dg_{y}( \res( \psi( f), \psi( g))) < (p+q)^k d$.

\begin{proposition} \label{prop:multi-SR-fast-computation} 
  \cite{Reischert:subresultant:97}
  We can compute $\SRQ( f, g)$, any polynomial in $\SR( f, g)$, 
  and $\res( f, g)$ w.r.t. $x$ 
  in  $\sOB( q (p+q)^{k+1} d \tau)$.
\end{proposition}

\begin{lemma} \label{lem:multi-SR-computation}
  We can compute $\SR( f, g)$ in $\sOB( q (p+q)^{k+2} d \tau)$.
\end{lemma}

\begin{proof} %
  Every polynomial in $\SR( f, g)$ has coefficients of 
  magnitude bounded $2^{c \,(p+q)\tau}$, for a suitable constant $c$,
  assuming $\tau > \lg( d)$.
  Consider the map $\chi : \ZZ[y] \mapsto \ZZ$,
  where
  $ y \mapsto 2^{\lceil c \, (p+q)\tau \rceil} $,
  and let $\phi = \psi \,\circ\, \chi : \ZZ[y_1, y_2 \dots, y_k] \rightarrow \ZZ$.
  Then $\bitsize{ \phi( f)}, \bitsize{ \phi( g)} \leq c \, (p+q)^{k} \, d \, \tau$.
  Now apply Proposition \ref{pr:SR-computation}.

  In order to complete the computation we should recover the result from the computed
  sequence, that is to apply the inverse image of $\phi$. The cost of this computation
  (almost linear w.r.t. the output) is dominated; which is always the case.
\end{proof} 

\begin{theorem} \label{th:multi-SR-fast-evaluation}
  We can %
  evaluate $\SR( f, g )$ at $x = \rat{ a}$
  where $\rat{a} \in \QQ \cup \{ \infty \}$ and $\bitsize{ \rat{ a}} = \sigma$,
  in $\sOB( q (p+q)^{k+1} d \max\{ \tau, \sigma\})$.
\end{theorem}
\begin{proof}
  First we compute $\SRQ( f, g)$
  in  $\sOB( q (p+q)^{k+1} d \, \tau)$
  (Proposition \ref{prop:multi-SR-fast-computation}),
  and then we evaluate the sequence over $\rat{a}$, using binary segmentation.
  For the latter we need to bound the bitsize of the resulting polynomials.  

  The polynomials in $\SR( f, g)$ have total degree
  in $y_1, \dots, y_k$ bounded by $(p+q)\sum_{i=1}^{k}{d_i}$
  and coefficient bitsize bounded by $(p+q)\tau$.
  With respect to $x$, the polynomials in $\SR( f, g)$
  have degrees in $\OO( p)$, so
  substitution $x = \rat{a}$ yields values of size $\sO( p \sigma)$.
  After the evaluation we obtain polynomials in $\ZZ[y_1, \dots, y_k]$
  with bitsize bounded by 
  $\max\{ (p+q)\tau, p \sigma\} \leq (p + q)\max\{ \tau, \sigma\}$.
  
  Consider the map $\chi : \ZZ[y] \rightarrow \ZZ$,
  where
  $y \mapsto 2^{\lceil c \, (p+q) \max\{ \tau, \sigma\} \rceil}$, for a suitable constant $c$.
  Apply the map $\phi = \psi \circ \chi$ to $f,g$.  
  Now, $\bitsize{ \phi(f)}, \bitsize{ \phi( g)} \leq  c\, d\, (p+q)^{k} \max\{ \tau, \sigma\}$.  
  By Proposition \ref{pr:SR-fast-evaluation}, the evaluation
  costs $\sOB( q (p+q)^{k+1} d \max\{ \tau, \sigma\})$.
\end{proof}

We obtain the following, for bivariate $f, g \in (\ZZ[y])[x]$,
such that $\dg_x( f) = p$, $\dg_x( g) = q$, $\dg_y( f), \dg_y( g) \leq d$. 

\begin{corollary} \label{cor:biv-SR-computation}
  We compute $\SR( f, g)$ in
  $\sOB( p q (p + q)^2 d \tau) $.
  For any polynomial $\SR_j( f, g)$ in $\SR(f, g)$, 
  $\dg_{x}( \SR_j( f, g)) = \OO( \max\{p, q\})$,
  $\dg_{y}( \SR_j( f, g)) = \OO( \max\{p,q\} d)$,
  and also $\bitsize{ \SR_j( f, g)} = \OO( \max\{p, q\} \tau)$.
\end{corollary}

\begin{corollary} \label{cor-biv-SR-fast-computation}
  We compute $\SRQ( f, g)$, any polynomial in $\SR(f, g)$,
  and $\res( f, g)$ in 
  $\sOB( p q \max\{p,q\} d \tau)$.
\end{corollary}

\begin{corollary} \label{cor:biv-SR-fast-evaluation}
  We can compute $\SR( f, g \,;\, \rat{a})$, 
  where $\rat{a}\in\QQ \cup \{ \infty \}$ and $\bitsize{\rat{a}} = \sigma$,
  in
  $\sOB(p q \max\{p,q\} d \max\{ \tau, \sigma\})$.
  For the polynomials $\SR_j( f, g \,;\, \rat{a}) \in \ZZ[y]$,
  except for $f, g$,
  it holds $\dg_{y}( \SR_j( f, g \,;\, \rat{a})) = \OO( (p+q) d)$
  and $\bitsize{ \SR_j( f, g \,;\, \rat{a}) } = \OO( \max\{p,q\} \tau + \min\{p,q\} \sigma)$.
\end{corollary}

We now reduce the computation of the sign of $F \in \ZZ[x,y]$
over $(\alpha,\beta)\in\ALG^2$ to that over several points in $\QQ^2$.
Let $\dg_{x}( F) = \dg_{y}( F) = n_1$,  $\bitsize{F} =\sigma$ and 
$\alpha \cong (A, [\rat{a}_1, \rat{a}_2])$, 
$\beta  \cong (B, [\rat{b}_1, \rat{b}_2])$,
where  $A, B \in \ZZ[X]$, $\dg( A) =\dg( B) = n_2$, 
$\bitsize{A} = \bitsize{B} = \sigma$.
We assume $n_1 \leq n_2$, which is relevant below.
The pseudo-code is in Algorithm \ref{alg:sign-at-2}, see \cite{Sakkalis:algnum:89},
and generalizes the univariate case, e.g.~\cite{emt-lncs-2006,Yap2000}.
For $A$, resp.\ $B$, 
we assume that we know their values on $\rat{ a}_1, \rat{a}_2$,
resp.\ $\rat{ b}_1, \rat{b}_2$.

\begin{myalgorithm}[tbp]
  \dontprintsemicolon
  \linesnumbered
  \SetFuncSty{textsc}
  \SetKw{RET}{{\sc return}} 
  \SetKw{OUT}{{\sc output \ }} 
  \SetVline \KwIn{$ F \in \ZZ[x, y], 
    \alpha \cong ( A, [\rat{a}_1, \rat{a}_2]), %
    \beta  \cong ( B, [\rat{b}_1, \rat{b}_2])$} %
  \KwOut{ $\sign( F( \alpha, \beta))$}
  \caption{$\func{sign\_at}(F, \alpha, \beta)$}

  compute $\SRQ_x( A, F)$\;

  $L_1 \leftarrow \SR_x( A, F \,;\, \rat{a}_1)$, $V_1 \leftarrow \emptyset$ \;
  \nllabel{a:biv-signat-eval-a1}
  \lForEach{ $f \in L_1$ }{
    $V_1 \leftarrow \FuncSty{add}( V_1, \FuncSty{ sign\_at}( f, \beta))$ \;
    \nllabel{a:biv-signat-signat-beta}
  }\;
  
  $L_2 \leftarrow \SR_x( A, F \,;\, \rat{a}_2)$, $V_2 \leftarrow \emptyset$ \;
  \lForEach{ $f \in L_2$ }{
    $V_2 \leftarrow \FuncSty{add}( V_2, \FuncSty{ sign\_at}( f, \beta))$ \;
  }\;

  \RET $\left(\FuncSty{var}(V_1) -  \FuncSty{var}(V_2) \right) \cdot \sign( A'( \alpha))$ \;
  \label{alg:sign-at-2}
\end{myalgorithm}

\begin{theorem}\label{th:biv_sign_at}
  We compute the sign of polynomial $F(x,y)$ over $\alpha, \beta$
  in  $\sOB(n_1^2 \, n_2^3 \, \sigma)$.
\end{theorem}

\begin{proof}
  First, we compute $\SRQ_x( A, F)$, in $\sOB( n_1^2 n_2^2 \sigma)$
  (Corollary \ref{cor-biv-SR-fast-computation}), so as to evaluate 
  $\SR( A, F)$ on the endpoints of $\alpha$.
  
  We compute $\SR( A, F; \rat{a}_1)$.
  The first polynomial in the sequence is $A$ and notice that we already know 
  its value on $\rat{ a}_1$.
  This computation costs $\sOB( n_1^2 \, n_2^3 \, \sigma)$ by
  Corollary \ref{cor:biv-SR-fast-evaluation}
  with $q = n_1$, $p = n_2$, $d = n_1$, $\tau = \sigma$, and $\sigma = n_2 \sigma$,
  where the latter corresponds to the bitsize of the endpoints.
  After the evaluation we obtain a list $L_1$, 
  which contains $\OO( n_1)$ polynomials, say $f \in \ZZ[y]$, 
  such that $\dg( f) = \OO( n_1 n_2)$.
  To bound the bitsize, 
  notice that the polynomials in $\SR( f, g)$ 
  are of degrees $\OO( n_1)$ w.r.t. $x$ 
  and of bitsize $\OO( n_2 \sigma)$.
  After we evaluate on $\rat{a}_1$, $\bitsize{f} = \OO( n_1 n_2 \sigma)$.

  For each $f\in L_1$ we compute its sign over $\beta$
  and count the sign variations.  
  We could apply directly Corollary \ref{cor:sign-at-1},
  but we can do better. 
  If $\dg( f) \geq n_2 $ then
  $\SR( B, f) = $
  $(B, f, -B$, $g = -\prem{ f, -B}, \dots)$. We start the evaluations at $g$:
  it is computed in $\sOB( n_1^2 n_2^3 \sigma)$ (Proposition \ref{pr:SR-computation}),
  $\dg( g) = \OO( n_2)$ and $\bitsize{g} = \OO( n_1 n_2 \sigma)$.
  Thus,  we evaluate $\SR( -B, g; \rat{a}_1)$ in $\sOB( n_1 n_2^3 \sigma)$,
  by Corollary \ref{cor:sign-at-1}, 
  with $p = q = n_2$, $\tau_f = \sigma$, $\tau = n_1 n_2 \sigma$.
  If $\dg( f)<n_2$ the complexity is dominated.
  Since we perform $\OO( n_1)$ such evaluations,
  all of them cost $\sOB( n_1^2 n_2^3 \sigma)$.

  We repeat for the other endpoint of $\alpha$,
  subtract the sign variations,
  and multiply by $\sign( A'(\alpha)),$
  which is known from the process
  that isolated $\alpha$.
  If the last sign in the two sequences is alternating,
  then $\sign( F( \alpha, \beta)) = 0$.  
\end{proof}

\section{Bivariate real solving} \label{sec:biv-solving}

Let $F, G \in \ZZ[x, y]$,
$\dg( f) = \dg( g) = n$ and $\bitsize{F} = \bitsize{G} = \sigma$.
We assume relatively prime polynomials for simplicity 
but this hypothesis is not restrictive because it can be verified 
and, if it does not hold, it can be imposed within the same asymptotic complexity.
We study the problem of real solving
the system $F = G = 0$.
The main idea is to project the roots on their $x$-
and $y$-coordinates.
The difference between the algorithms is the
way they match %
coordinates.

\subsection{The \func{grid} algorithm} \label{sec:grid-solve}

Algorithm \func{grid} %
is straightforward, see also \cite{et-casc-2005,WolPhd}.
The pseudo-code is in Algorithm \ref{alg:grid-solve}.
We compute the $x$- and
$y$-coordinates of the real solutions %
by solving
resultants $\res_x( F, G)$, %
$\res_y( F, G)$.
We match them using the algorithm \func{sign\_at} (Theorem \ref{th:biv_sign_at}) %
by testing all rectangles in this grid.

\begin{myalgorithm}[Htbp]
  \dontprintsemicolon
  \linesnumbered
  \SetFuncSty{textsc}
  \SetKw{RET}{{\sc return}} 
  \SetKw{OUT}{{\sc output \ }} 
  \SetVline \KwIn{$F, G \in \ZZ[x, y]$}
  \KwOut{ The real solutions of $F = G = 0$}

  \BlankLine
  
  $R_x \leftarrow \res_y( F, G)$ \;
  \nllabel{a:grid-Rx}
  $L_x, M_x \leftarrow \FuncSty{solve}( R_x)$ \;
  \nllabel{a:grid-Xalg}
  \BlankLine

  $R_y \leftarrow \res_x( F, G)$ \;
  \nllabel{a:grid-Ry}
  $L_y, M_y \leftarrow \FuncSty{solve}( R_y)$ \;
  \nllabel{a:grid-Yalg}

  \BlankLine
  
  $Q \leftarrow \emptyset$ \;
  \ForEach{ $\alpha \in L_x$ }{
    \ForEach{ $\beta \in L_y$ }{
      \lIf{ $\FuncSty{sign\_at}(F, \alpha, \beta) = 0 \,\wedge\,
        \FuncSty{sign\_at}(G, \alpha, \beta) = 0$}{
        \nllabel{a:grid-signat}
        $Q \leftarrow \FuncSty{add}( Q, \{ \alpha, \beta \})$\;
      }\;
    }\;
  }\;
  
  \RET $Q$\;
  \caption{$\func{grid}$($F, G$)}
  \label{alg:grid-solve}
\end{myalgorithm}

To the best of
our knowledge, this is the first time that the algorithm's complexity is studied.
Its simplicity makes it attractive;
however, \func{sign\_at} (Algorithm \ref{alg:sign-at-2}) is very costly.
The algorithm requires no genericity assumption on the input;
we study a generic shear that brings the system to generic position in order
to compute the multiplicities within the same complexity bound.
The algorithm allows the use of heuristics,
such as bounding the number of roots,
e.g.~Mixed Volume, or counting the roots with %
given
abscissa %
by Lemma \ref{lem:count-alpha}.

\begin{theorem} \label{th:grid_solve}
  Isolating all real roots of system $F = G = 0$ using \func{grid} 
  has complexity $\sOB(n^{14}+n^{13}\sigma)$, provided $\sigma=O(n^3)$;
  or in $\sOB(N^{14})$, where $N = \max\{n, \sigma\}$.
\end{theorem}
\begin{proof}
  We
  compute %
  resultant 
  $R_x$ of $F, G$ w.r.t.~$y$
  (line~\ref{a:grid-Rx} in Algorithm \ref{alg:grid-solve}).
  The complexity is 
  $\sOB(n^4 \sigma)$, using Corollary \ref{cor-biv-SR-fast-computation},
  with $p = q = d = n$ and $\tau = \sigma$.
  Notice that $\dg( R_x) = \OO(n^2)$, $\bitsize{ R_x} = \OO(n\, \sigma)$.
  We isolate its real roots
  in $\sOB(n^{12} + n^{10} \sigma^2)$ (Proposition \ref{pr:solve-1})
  and store them in $L_x$.
  This complexity shall be dominated. %
  We do the same for the $y$ axis
  (lines~\ref{a:grid-Ry} and \ref{a:grid-Yalg} in Algorithm \ref{alg:grid-solve})
  and store the roots in $L_y$.

  The representation of the %
  algebraic numbers %
  contains the square-free part of $R_x$ or $R_y$,
  which has
  the bitsize %
  $\OO( n^2 + n\, \sigma)$ \cite{BPR06,emt-lncs-2006}.
  The isolating intervals have endpoints of bitsize $\OO( n^4 + n^3 \, \sigma)$.
  Let $r_x$, %
  $r_y$ be the 
  number of real roots of the corresponding %
  resultant,
  both in
  $\OO( n^2)$.
  For every pair of
  algebraic numbers from $L_x$
  and $L_y$,
  we test whether $F, G$ vanish
  using \func{sign\_at} (Theorem \ref{th:biv_sign_at} and Algorithm \ref{alg:sign-at-2}).
  Each %
  test
  costs $\sOB(n^{10} + n^{9}\sigma)$
  and we perform $r_x \, r_y = O(n^4)$ of them.
\end{proof}

We now examine the multiplicity of a root $(\alpha,\beta)$ of the system.
Refer to \cite[Section II.6]{BrKn} %
for its definition as the exponent of factor $(\beta x-\alpha y)$
in the resultant of the (homogenized) polynomials, under certain assumptions.
Previous work includes \cite{VegKah:curve2d:96,sakkalis90:alg-curves,WolSei:topology:05}.
Our algorithm reduces to bivariate sign determination and
does not require bivariate factorization.  
The sum of multiplicities of all roots $(\alpha,\beta_j)$
equals the multiplicity of $x=\alpha$ in the respective resultant.
We apply a shear transform
so as to ensure
that different roots project to different points on the $x$-axis.  

\subsubsection{Deterministic shear}\label{sec:determ-shear}
We determine an adequate (horizontal) shear %
such that
\begin{equation}\label{EshearedRes}
  R_t(x) = \res_y \left( F(x+ty,y), G(x+ty,y) \right) , 
\end{equation}
has simple roots corresponding to the projections
of the common roots of the system $F(x,y)=G(x,y)=0$, when $t\mapsto t_0\in\ZZ$,
and the degree of the polynomials remains the same.
Notice that this shear does not affect inherently multiple roots, which exist
independently of the reference frame.
$R_{red}\in (\ZZ[t])[x]$ is the squarefree part of the resultant,
as an element of UFD $(\ZZ[t])[x]$, and its discriminant, with respect to $x$, is
$\Delta\in\ZZ[t]$.
Then $t_0$ must be such that $\Delta(t_0)\ne 0$.

\begin{lemma} \label{lem:shear-value}
  Computing $t_0\in\ZZ$, such that the corresponding shear is
  sufficiently generic, has complexity $\sOB(n^{10} + n^{9}\sigma)$.
\end{lemma}

\begin{proof} %
  Suppose $t_0$ is such that the degree does not change.
  It suffices to find, among $n^4$ integer numbers, one that does not
  make $\Delta$ vanish; note that all candidate values are of bitsize $\OO(\log n)$.

  We perform the substitution $(x, y) \mapsto (x + t y, y)$ to $F$ and
  $G$ and compute the resultant w.r.t.\ $y$ 
  in $\sOB( n^5  \sigma)$, which lies in $\ZZ[t, x]$,
  of degree $\OO( n^2)$ and bitsize $\sO( d \sigma)$ (Proposition \ref{prop:multi-SR-fast-computation}).
  We consider this polynomial as univariate in $x$ and compute
  its square-free part, and then the discriminant of its
  square-free part. Both operations cost $\sOB( n^{10} + n^9 \sigma)$
  and the discriminant is a polynomial in $\ZZ[t]$ of degree 
  $\OO( n^4)$ and bitsize $\sO(d^4 + d^3 \sigma)$ (Corollary \ref{cor-biv-SR-fast-computation}).

  We can evaluate the discriminant over all the first $n^4$ positive integers,
  in $\sOB( n^8 + n^3 \sigma)$, using the multipoint evaluation algorithm, 
  see  %
  e.g.~\cite{vzGGer}.
  Among these integers, there is at least one that is not a root of
  the discriminant.
\end{proof}

The idea here is to use explicit candidate values of $t_0$ right from the
start.
In practice, the above complexity becomes $\sOB(n^{5} \sigma)$, because a
constant number of tries or a random value will typically suffice.
For an alternative approach, see \cite{VegNec:topology:02}, and \cite{BPR06}.
It is straightforward to compute the multiplicities of the
sheared system.
Then, we need to match the latter with the roots of the original system, which is nontrivial in practice.

\begin{theorem} \label{th:grid-multiplicities}
  Consider the setting of Theorem \ref{th:grid_solve}.
  Having isolated all real roots of $F = G = 0$, it is possible to
  determine their multiplicities in
  $\sOB(n^{12} + n^{11} \sigma + n^{10} \sigma^2)$.
\end{theorem}
\begin{proof}
  By the previous lemma, $t\in\ZZ$ is determined, with $\bitsize t = \OO(\log n)$,
  in $\sOB(n^{10} + n^{9}\sigma)$.
  Using this value, we isolate all the real roots of $R_t(x)$, defined in~(\ref{EshearedRes}),
  and determine their multiplicities in $\sOB( n^{12} + n^{10} \sigma^2 )$ (Proposition \ref{pr:solve-1}).  
  Let $\rho_j \simeq ( R_t(x), [r_j ,r_j'] )$ be the real roots, for $j=0,\dots,r-1$.

  By assumption, we have already isolated the roots of the system, denoted by
  $(\alpha_i,\beta_i)\in [a_i,a_i']\times [b_i,b_i']$, where $a_i,a_i',b_i,b_i'\in\QQ$
  for $i=0,\dots,r-1$.
  It remains to match each pair $(\alpha_i,\beta_i)$ to a unique $\rho_j$ by
  determining function $\phi:\{0,\dots,r-1\}\rightarrow\{0,\dots,r-1\}$,
  such that $\phi(i)=j$ iff $(\rho_j,\beta_i)\in\ALG^2$ is a root of the sheared system
  and $\alpha_i=\rho_j+t\beta_i$.

  Let $[c_i,c_i']= [a_i,a_i'] - t[b_i,b_i'] \in\QQ^2$.  
  These intervals may be overlapping.
  Since the endpoints have bitsize $\OO(n^4 + n^3\sigma)$, the intervals $[c_i,c_i']$
  are sorted in $\sOB(n^6 + n^5\sigma)$.
  The same complexity bounds the operation of merging this interval list with
  the list of intervals $[r_j,r_j']$.
  If there exist more than one $[c_i,c_i']$ overlapping with some $[r_j,r_j']$,
  some subdivision steps are required so that the intervals reach the bitsize of
  $s_j$, where $2^{s_j}$ bounds the separation distance associated to the $j$-th root.
  By Proposition \ref{lem:very-important}, $\sum_i s_i= \OO(n^4 + n^3 \sigma)$.

  Our analysis resembles that of~\cite{emt-lncs-2006} for proving Proposition \ref{pr:solve-1}.
  The total number of steps is $\OO(\sum_i s_i)= \OO(n^4 + n^3 \sigma)$,
  each requiring an evaluation of $R(x)$ over an endpoint of size $\le s_i$.
  This evaluation costs $\sOB(n^4 s_i)$, leading to an overall cost of
  $\sOB(n^8 + n^7\sigma)$ per level of the tree of subdivisions.
  Hence, the overall complexity is bounded by $\sOB(n^{12} + n^{11} \sigma + n^{10} \sigma^2)$.
\end{proof}

\subsection{The \func{m\_rur} algorithm}
\func{m\_rur} assumes that the polynomials are in Generic Position:
different roots project to different $x$-coordinates and leading coefficients
w.r.t.~$y$ have no common real roots.
\begin{proposition}   \label{pr:SR-rur}
  \cite{VegKah:curve2d:96,BPR06}
  Let $F, G$ be co-prime polynomials, in generic position.
  If $\SR_j(x, y) = \sr_{j}(x) y^j + \sr_{j, j-1}(x) y^{j-1}$ $+ \dots + \sr_{j, 0}(x)$,
  and $(\alpha, \beta)$ is a real solution of the system $F = G = 0$,
  then there exists $k$, such that
  $\sr_0(\alpha) = \dots = \sr_{k-1}(\alpha) = 0$,
  $\sr_k(\alpha) \neq 0$ and 
  $\beta = -\frac{1}{k}\frac{\sr_{k, k-1}(\alpha)}{\sr_k(\alpha)}$.
\end{proposition}

This expresses the ordinate of a solution
in a Rational Univariate Representation (RUR) of the abscissa.  
The RUR applies to multivariate algebraic systems
\cite{Ren89,Can88pspace,Rou:rur:99,BPR06};
by generalizing
the primitive element method %
by Kronecker.
Here we adapt it to small-dimensional systems.

Our algorithm is similar to 
\cite{VegNec:topology:02,VegKah:curve2d:96}.
However, their algorithm
computes only a RUR using Proposition \ref{pr:SR-rur}, so
the representation of the ordinates remains implicit.
Often, this representation is not sufficient
(we can always compute the minimal polynomial of the roots, 
but this is highly inefficient).
We modified the algorithm \cite{et-casc-2005},
so that the output includes isolating rectangles,
hence the name modified-RUR (\func{m\_rur}).  
The most important difference with
\cite{VegKah:curve2d:96} is that they represent 
algebraic numbers by Thom's encoding %
while we use isolating intervals, which were thought of having high theoretical complexity. 

The pseudo-code of \func{m\_rur} is in Algorithm \ref{alg:mrur-solve}.  
We project on the $x$ and the $y$-axis; for each real solution on the $x$-axis 
we compute its ordinate using Proposition \ref{pr:SR-rur}.  
First we compute the sequence $\SR( F, G)$ w.r.t. $y$ 
in $\sOB(n^5 \, \sigma)$ (Corollary \ref{cor:biv-SR-computation}).  

\begin{myalgorithm}[tbp]
  \dontprintsemicolon
  \linesnumbered
  \SetFuncSty{textsc}
  \SetKw{RET}{{\sc return}} 
  \SetKw{OUT}{{\sc output \ }} 
  \SetVline \KwIn{$F, G \in \ZZ[X, Y]$ in generic position}
  \KwOut{ The real solutions of the system $F = G = 0$}

  \caption{\func{m\_rur}($F, G$)}

  $\SR \leftarrow \SR_y( F, G)$ \;                    \nllabel{a:mrur-SR}

  \tcc{Projections and real solving with multiplicities}
  $R_x \leftarrow \res_y( F, G)$ \;                   \nllabel{a:mrur-Rx}
  $P_x, M_x \leftarrow \FuncSty{solve}( R_x)$ \;  \nllabel{a:mrur-Xalg}

  $R_y \leftarrow \res_x( F, G)$ \;                       \nllabel{a:mrur-Ry}
  $P_y, M_y \leftarrow \FuncSty{solve}( R_y)$ \;      \nllabel{a:mrur-Yalg}
  $I \leftarrow \FuncSty{intermediate\_points}( P_y)$ \;  \nllabel{a:mrur-IntPt}

  \tcc{Factorization of $R_x$ according to $\sr$ }
   $K \leftarrow \FuncSty{compute\_k}(\SR, P_x)$ \; \nllabel{a:mrur-k}

  $Q \leftarrow \emptyset$ \;
  \tcc{Matching the solutions}
  \ForEach{ $\alpha \in P_x$ }{       \nllabel{a:mrur-loop}

    $\beta \leftarrow \FuncSty{find}( \alpha, K, P_y, I)$ \;
    \nllabel{a:mrur-find}

    $Q \leftarrow \FuncSty{add}( Q, \{ \alpha, \beta \})$\;
    \nllabel{a:mrur-add-sol}
  }\;
  
  \RET $Q$\;
  \label{alg:mrur-solve}
\end{myalgorithm}

\paragraph*{Projection.}
This is similar to \func{grid}.
The complexity is dominated by real solving the resultants, 
i.e. $\sOB( n^{12} + n^{10} \, \sigma^2)$.
Let $\alpha_i$, resp.~$\beta_j$, be the real root coordinates.
We compute rationals $q_j$ 
between the $\beta_j$'s in $\sOB( n^5 \sigma)$,
viz. $\func{intermediate\_points}( P_y)$;
the $q_j$ have aggregate bitsize $\OO( n^3\, \sigma)$ (Lemma \ref{lem:very-important}):
\begin{equation}
  q_0 < \beta_1 < q_1 < \beta_2 < \dots < 
  \beta_{\ell-1} < q_{\ell-1} < \beta_\ell <  q_{\ell} , \label{eq:y-sols}
\end{equation}
where $\ell \leq 2\, n^2$.
Every $\beta_j$ corresponds to a unique $\alpha_i$.
The multiplicity of $\alpha_i$ as a root of $R_x$ is 
the multiplicity of a real solution of the system, that has it as abscissa.

\paragraph*{Sub-algorithm} \func{compute\_k}.
In order to apply Proposition \ref{pr:SR-rur}, for every $\alpha_i$ we
must compute $k\in\NN^*$ such that the assumptions of the theorem are fulfilled;
this is possible by genericity.
We follow \cite{mpsttw-aicg-06,VegKah:curve2d:96} and
define recursively polynomials $\Gamma_{j}(x)$:
Let $\Phi_{0}(x) = \frac{\sr_{0}(x)}{\pgcd(\sr_{0}(x), \sr_{0}'(x))}$,
$\Phi_{j}(x) = \pgcd(\Phi_{j-1}(x), \sr_{j}(x))$, and  $\Gamma_{j} = \frac{\Phi_{j-1}(x)}{\Phi_{j}(x)}$,
for $j > 0$.
Now $\sr_i(x) \in \ZZ[x]$ is the principal subresultant coefficient of
$\SR_i \in (\ZZ[x])[y]$, and $\Phi_{0}(x)$ is the square-free part of $R_x = \sr_{0}(x)$.
By construction, $\Phi_0(x) = \prod_j{ \Gamma_j}(x)$
and $\pgcd( \Gamma_j, \Gamma_i) = 1$, if $j \ne i$.
Hence every $\alpha_i$ is a root of a unique $\Gamma_j$ and the latter
switches sign at the interval's endpoints. %
Then, $\sr_0(\alpha) = \sr_{1}(\alpha) = 0, \dots, \sr_{j}(\alpha) =0$,
$\sr_{j+1}(\alpha) \neq 0$; thus $k = j+1$.

It holds that $\dg( \Phi_0) = \OO( n^2)$ and
$\bitsize{ \Phi_0} = \OO( n^2 + n\,\sigma)$.
Moreover, $\sum_j{ \dg( \Gamma_j)} = \sum_j{\delta_j} = \OO( n^2)$ 
and, by Mignotte's bound \cite{MignotteStefanecu},
$\bitsize{ \Gamma_j} = \OO( n^2 + n \sigma)$.
To compute the factorization $\Phi_0(x) = \prod_j{ \Gamma_j}(x)$ 
as a product of the $\sr_j(x)$,
we perform $\OO( n)$ gcd computations of polynomials of degree 
$\OO(n^2)$ and bitsize $\sO( n^2 + n \sigma)$.
Each gcd computation costs $\sOB( n^6 + n^5 \, \sigma)$ (Proposition \ref{pr:SR-computation})
and thus the overall cost is $\sOB( n^7 + n^6 \, \sigma)$.

We compute the sign of the $\Gamma_j$ over all the $O(n^2)$ isolating endpoints
of the $\alpha_i$, which have aggregate bitsize 
$\OO(n^4 + n^3 \, \sigma)$ (Lemma \ref{lem:very-important})
in $\sOB( \delta_j n^4  + \delta_j n^3 \sigma + \delta_j^2 (n^4 + n^3 \sigma))$,
using Horner's rule.
Summing over all $\delta_j$, the complexity is 
$\sOB( n^8 + n^7 \sigma)$.
Thus the overall complexity is $\sOB( n^9 + n^8 \, \sigma)$.  

\paragraph*{Matching and algorithm} \func{find}.
The process takes a real root of $R_x$ and computes the ordinate $\beta$
of the corresponding root of the system.
For some real root $\alpha$ of $R_x$ we represent the ordinate
$  A(\alpha) = -\frac{1}{k} \frac{\sr_{k,k-1}(\alpha)}{\sr_{k}(\alpha)} 
  =\frac{A_1(\alpha)}{A_2(\alpha)} .$
The generic position assumption guarantees that there is a unique 
$\beta_j$, in $P_y$, such that 
$\beta_j = A( \alpha)$, where $1 \leq j \leq \ell$.
In order to compute $j$ we use~(\ref{eq:y-sols}):
$  q_j <  A( \alpha) = \frac{A_1(\alpha)}{ A_2(\alpha)} = \beta_j < q_{j + 1}$.
Thus $j$ can be computed by binary search in $\OO( \lg{\ell}) = \OO( \lg{n})$
comparisons of $A( \alpha)$ with the $q_j$.
This is equivalent to computing the sign of 
$B_j( X) = A_1( X) -  q_j \, A_2(X)$ over $\alpha$ by executing
$\OO( \lg{ n})$ times, $\func{sign\_at}( B_j,  \alpha)$.

Now, $\bitsize{q_j} = \OO( n^4 + n^3 \sigma)$ 
and $\dg( A_1) = \dg(\sr_{k,k-1}) = \OO(n^2)$, 
$\dg(A_2) = \dg( \sr_k) = \OO(n^2)$,
$\bitsize{ A_1} = \OO(n \,\sigma)$, 
$\bitsize{ A_2} = \OO(n \, \sigma)$.
Thus $\dg( B_j) = \OO( n^2)$ and $\bitsize{ B_j} = \OO( n^4 + n^3\, \sigma)$.  
We conclude that $\func{sign\_at}( B_j,  \alpha)$ 
and \func{find} have complexity $\sOB( n^8 + n^7 \sigma)$
(Corollary \ref{cor:sign-at-1}).
As for the overall complexity of the loop
(Lines \ref{a:mrur-loop}-\ref{a:mrur-add-sol}) the complexity is
$\sOB( n^{10} + n^9 \sigma)$,
since it is executed $\OO( n^2)$ times.

\begin{theorem} \label{th:mrur-solve}
  We isolate all real roots of $F = G = 0$, 
  if $F$, $G$ are in generic position, 
  by \func{m\_rur} in $\sOB( n^{12} + n^{10} \sigma^2)$;
  or in $\sOB(N^{12})$, where $N = \max\{n, \sigma\}$.
\end{theorem}

The generic position assumption is without loss of generality since
we can always put the system in such position by applying a
shear transform; see Section \ref{sec:determ-shear} and also~\cite{BPR06,VegKah:curve2d:96}.
The bitsize of polynomials of the (sheared) system 
becomes $\sO( n + \sigma )$ \cite{VegKah:curve2d:96}
and does not change the bound of Theorem \ref{th:mrur-solve}.
However, now is raised the problem of expressing
the real roots in the original coordinate system
(see the proof of Theorem \ref{th:grid-multiplicities}).

\subsection{The \func{g\_rur} algorithm} 
\label{sec:grur}
In this section we present an algorithm that uses some ideas from \func{m\_rur} 
but also relies on GCD computations of polynomials with coefficients in an extension field
to achieve efficiency (hence the name \func{g\_rur}).
The pseudo-code of \func{g\_rur} is in Algorithm \ref{alg:grur-solve}.
For GCD computations with polynomials with coefficients in an extension field we use the
algorithm, and the \maple implementation, of \citet{HoeMon:gcd:02}. 

\begin{myalgorithm}[tbp]
  \dontprintsemicolon
  \linesnumbered
  \SetFuncSty{textsc}
  \SetKw{RET}{{\sc return}} 
  \SetKw{OUT}{{\sc output \ }} 
  \SetVline \KwIn{$F, G \in \ZZ[X, Y]$}
  \KwOut{ The real solutions of the system $F = G = 0$}

  \caption{\func{g\_rur}($F, G$)}

  \tcc{Projections and real solving with multiplicities}
  $R_x \leftarrow \res_y( F, G)$ \;                   \nllabel{a:grur-Rx}
  $P_x, M_x \leftarrow \FuncSty{solve}( R_x)$ \;  \nllabel{a:grur-Xalg}

  $R_y \leftarrow \res_x( F, G)$ \;                       \nllabel{a:grur-Ry}
  $P_y, M_y \leftarrow \FuncSty{solve}( R_y)$ \;      \nllabel{a:grur-Yalg}
  \tcc{$I$ contains the rationals $q_1 < q_2 < \cdots < q_{|I|}$}
  $I \leftarrow \FuncSty{intermediate\_points}( P_y)$ \;  \nllabel{a:grur-IntPt}

  $Q \leftarrow \emptyset$ \;
  \ForEach{ $\alpha \in P_x$ }{       \nllabel{a:grur-loop}
    $\overline{F} \leftarrow \FuncSty{SquareFreePart}(F(\alpha, y))$ \;
    $\overline{G} \leftarrow \FuncSty{SquareFreePart}(G(\alpha, y))$ \;
    $H \leftarrow \FuncSty{gcd}(\overline{F}, \overline{G}) \in (\ZZ[\alpha])[y]$ \;
    \For{$j\leftarrow 1$ \KwTo $|I| - 1$}{
        \If{$H(\alpha, q_j)\cdot H(\alpha, q_{j+1}) < 0$}{
            \tcc{$P_{y}[j]$ indicates the $j$-th element of $P_y$}
            $Q \leftarrow \FuncSty{add}(Q, \{\alpha, P_y[j]\})$ \;
        }
    }
  }\;
  
  \RET $Q$\;
  \label{alg:grur-solve}
\end{myalgorithm}

The first steps are similar to the previous algorithms:
We project on the axes, we perform real solving 
and compute the intermediate points on the $y$-axis.
The complexity is $\sOB( n^{12} + n^{10} \sigma^2)$.

For each $x$-coordinate, say $\alpha$,
we  compute the square-free part of $F( \alpha, y)$ and $G( \alpha, y)$,
say $\bar{F}$ and $\bar{G}$.
The complexity is that of computing the gcd with the derivative.
In  \cite{HoeMon:gcd:02} the cost is 
$\sOB(m M N D + m N^2 D^2 + m^2 k D)$, where
$M$ is the bitsize of the largest coefficient,
$N$ is the degree of the largest polynomial,
$D$ is the degree of the extension, 
$k$ is the degree of the gcd,  
and $m$ is the number of primes needed.
This bound does not assume fast multiplication algorithms,
thus, under this assumption, it becomes
$\sOB(m M N D + m N D + m k D)$.

In our case 
$M = \OO( \sigma)$, $N = \OO( n)$, $D = \OO( n^2)$,  
$k = \OO( n)$, and $m = \OO( n \sigma)$.
The cost is $\sOB( n^4 \sigma^2)$ and since we repeat it $\OO( n^2)$ times,
the overall cost is $\sOB( n^6 \sigma^2)$.
Notice the bitsize of the result is $\sOB( n + \sigma)$ \cite{BPR06}.

Now for each $\alpha$, we compute $H = \gcd( \bar{F}, \bar{G})$.
We have
$M = \OO( n + \sigma)$, $N = \OO( n)$, $D = \OO( n^2)$,  
$k = \OO( n)$, and $m = \OO( n^2 + n \sigma)$,
so the cost of each operation is $\sOB( n^6 + n^4 \sigma^2)$
and overall $\sOB( n^8 + n^6 \sigma^2)$.
The size of $m$ comes from Mignotte's bound \cite{MignotteStefanecu}.
$H$ is a square-free polynomial in $(\ZZ[ \alpha])[y]$,
of degree $\OO( n)$ and bitsize $\OO( n^2 + n \sigma)$, whose real roots correspond
to the real solutions of the system with abscissa $\alpha$.
The crux of the method is that
$H$ changes sign only over the intervals that contain its real roots.
To check these signs, it suffices to substitute $y$ in $H$ by the intermediate points,
thus obtaining a polynomial in $\ZZ[\alpha]$,
of degree $\OO( n)$ and bitsize  $\OO( n^2 + n \sigma + n s_j)$, 
where $s_j$ is the bitsize of the $j$-th intermediate point.

Now, we consider this polynomial in $\ZZ[x]$ and evaluate it over $\alpha$.
Using Corollary \ref{cor:sign-at-1}
with $p = n^2$, $\tau_f = n^2 + n \sigma$, $q = n$, and $\tau_g = n^2 + n \sigma + n s_j$,
this costs $\sOB( n^6 + n^5 \sigma + n^4 s_j)$.
Summing over $\OO( n^2)$ points and using
Lemma \ref{lem:very-important}, we obtain $\sOB( n^8 + n^7 \sigma)$.
Thus, the overall complexity is $\sOB( n^{10} + n^9 \sigma)$.

\begin{theorem} \label{th:grur}
  We can isolate the real roots of the system $F = G = 0$, 
  using \func{g\_rur} in $\sOB( n^{12} + n^{10} \sigma^2)$;
  or $\sOB(N^{12})$, where $N = \max\{n, \sigma\}$.
\end{theorem}

\section{Applications}

\subsection{Real root counting}
Let $F \in \ZZ[x, y]$, such that $\dg_{x}( F) =\dg_{y}( F) = n_1$ and $\bitsize{F} = \sigma$.
Let $\alpha, \beta\in \ALG$, such that 
$\alpha \cong ( A, [\rat{a}_1, \rat{a}_2])$ and $\beta \cong ( B, [\rat{b}_1, \rat{b}_2])$,
where $\dg( A), \dg( B)=n_2, \bitsize{A},\bitsize{B}\le \tau$
and $\rat{ c} \in \QQ$, such that $\bitsize{ \rat{ c}} = \lambda$.
Moreover, assume that $n_1^2 = \OO ( n_2 )$, as is the case in applications.
We want to count the number of real roots of 
$\bar{F} = F( \alpha, y) \in (\ZZ( \alpha))[y]$
in $(-\infty, +\infty)$, in $(\rat{ c}, +\infty)$ and in $( \beta, +\infty)$.  
We may assume that the leading coefficient of $\bar{ F}$ is nonzero. 
This is w.l.o.g.\ since we can easily check it,
and/or we can use the good specialization properties 
of the subresultants \cite{LickteigRoy:FastCauchy:01,VegLomRecRoy:StHa:89,VegKah:curve2d:96}.

Using Sturm's theorem, e.g. \cite{BPR06,Yap2000}, the number of real
roots of $\bar{ F}$ is  
$\var( \SR( \bar{ F}, \bar{ F}_y ; -\infty)) - \var( \SR( \bar{ F}, \bar{ F}_y ; +\infty))$.
Hence, we have to compute the sequence $\SR( \bar{F}, \bar{F}_y)$ w.r.t. $y$,
and evaluate it on $ \pm \infty$
or, equivalently, to compute the signs of the principal subresultant coefficients,
which lie in $\ZZ( \alpha)$.
This procedure is equivalent, due to the good specialization properties
of subresultants \cite{BPR06,VegLomRecRoy:StHa:89},
to computing the principal subresultant coefficients of $\SR( F, F_y)$,
which are polynomials in $\ZZ[x]$, and to evaluate them over $\alpha$.
In other words, the good specialization properties assure us that we
can compute a nominal sequence  
by considering the bivariate polynomials, and then perform the 
substitution $x = \alpha$.

The sequence $\sr$ of the principal subresultant coefficients
can be computed in $\sOB( n_1^4 \sigma)$, using Corollary \ref{cor-biv-SR-fast-computation}
with $p = q = d = n_1$, and $\tau = \sigma$. 
Now, $\sr$ contains $\OO( n_1)$ polynomials in $\ZZ[x]$,
each of degree $\OO( n_1^2)$ and bitsize $\OO( n_1 \sigma)$.
We compute the sign of each one evaluated over $\alpha$ in 
$$\sOB( n_1^2 n_2 \max\{ \tau, n_1 \sigma\} + n_2 \min\{ n_1^2, n_2\}^2 \tau)$$
using Corollary \ref{cor:sign-at-1} with 
$p = n_2$, $q = n_1^2$,  $\tau_f = \tau$,  and $\tau_g = n_1 \sigma$.
This proves the following:

\begin{lemma} \label{lem:count-alpha}
  We count the number of real roots of $\bar{ F}=F(\alpha,y)$ in 
  $\sOB( n_1^4 n_2 \sigma + n_1^5  n_2 \tau)$.
\end{lemma}

In order to count the real roots of $\bar{ F}$ in $(\beta, +\infty)$,
we use again Sturm's theorem. 
The complexity of the computation is dominated by 
the cost of computing $\var( \SR( \bar{ F}, \bar{ F}_y ; \beta))$,
which is equivalent to computing $\SR( F, F_y)$ w.r.t. to $y$,
which contains bivariate polynomials, and to compute their signs over $( \alpha, \beta)$.
The cost of computing $\SR( F, F_y)$ is $\sOB(n_1^5 \sigma)$
using Corollary \ref{cor:biv-SR-computation} with $p = q = d = n_1$, and $\tau = \sigma$. 
The sequence contains $\OO(n_1)$ polynomials in $\ZZ[x, y]$
of degrees $\OO( n_1)$ and $\OO( n_1^2)$, w.r.t. $x$ and $y$  respectively,
and bitsize $\OO(n_1 \sigma)$.
We compute the sign of each over $(\alpha, \beta)$
in $\sOB( n_1^4 n_2^3 \max\{ n_1 \sigma, \tau\})$ (Theorem \ref{th:biv_sign_at}).
This proves the following:
 
\begin{lemma} \label{lem:count-alpha-beta}
  We count the number of real roots of $\bar{F}$ in $(\beta, +\infty)$ 
  in $\sOB( n_1^5 n_2^3 \max\{ n_1 \sigma, \tau\})$.
\end{lemma}

By a more involved analysis, taking into account
the difference in the degrees of the bivariate polynomials, we can gain a factor. 
We omit it for reasons of simplicity.  
Finally, in order to count the real roots of $\bar{ F}$ in $( \rat{ c}, +\infty)$,
it suffices to evaluate the sequence $\SR( F, F_y)$ w.r.t. $y$ on $\rat{c}$,
thus obtaining polynomials in $\ZZ[x]$, and compute 
their signs
over $\alpha$.

The cost of the evaluation $\SR( F, F_y ; \rat{c})$ is
$\sOB( n_1^4 \max\{ \sigma, \lambda\})$,
using Corollary \ref{cor:biv-SR-fast-evaluation}
with $p = q = d = n_1$, $\tau = \sigma$ and $\sigma = \lambda$.
The evaluated sequence contains $\OO( n_1)$ polynomials in $\ZZ[x]$, 
of degree $\OO(n_1^2)$ and bitsize $\OO( n_1 \max\{ \sigma, \lambda\})$.
The sign of each one evaluated over $\alpha$ can be 
computed 
in 
$$\sOB(n_1^2 n_2 \max\{ \tau, n_1 \sigma, n_1 \lambda\} + n_1^4 n_2 \tau),$$
using Corollary \ref{cor:sign-at-1} with
$p = n_2$, $q = n_1^2$, $\tau_f = \tau$ and $\tau_g = n_1\max\{ \sigma, \lambda\}$.
This leads to the following:

\begin{lemma} \label{lem:count-c}
  We count the number of real roots of $\bar{F}$ in $( \rat{c}, +\infty)$
  in $\sOB(n_1^4 n_2 \max\{ n_1\tau, \sigma, \lambda\} )$.
\end{lemma} 

\subsection{Simultaneous inequalities in two variables}
Let 
$P, Q$, $A_1, \dots, A_{\ell_1}$, $B_1, \dots, B_{\ell_2}$, 
$C_1, \dots, C_{\ell_3} \in \ZZ[X, Y]$,
such that their total degrees are bounded by $n$ 
and their bitsize by  $\sigma$. 
We wish to compute $(\alpha, \beta) \in \ALG^2$
such that $P(\alpha, \beta) = Q(\alpha, \beta) = 0$
and also $A_i(\alpha, \beta) > 0$, $B_j(\alpha, \beta) < 0$ and 
$C_k(\alpha, \beta) = 0$, where
$1 \leq i \leq \ell_1, 1 \leq j \leq \ell_2, 1 \leq k \leq \ell_3$.
Let $\ell = \ell_1 + \ell_2 + \ell_3$.
\begin{corollary} \label{cor:biv-inequalities}
  There is an algorithm that solves the problem of $\ell$ simultaneous
  inequalities of degree $\le n$ and bitsize $\le \sigma$,
  in $\sOB( \ell n^{12} + \ell n^{11} \sigma + n^{10} \sigma^{2})$.
\end{corollary}
\begin{proof}
  Initially we compute the isolating interval representation of the real roots of 
  $P= Q = 0$ in $\sOB(n^{12} + n^{10} \sigma^2)$, using \func{grur\_solve}.
  There are  $\OO(n^2)$ real solutions,
  which are represented in isolating interval representation,
  with polynomials of degrees $\OO( n^2)$ and bitsize $\OO( n^2 + n \sigma)$.
  
  For each real solution, say $(\alpha, \beta)$,
  for each polynomial $A_i$, $B_j$, $C_k$ we compute the signs of
  $\sign(A_i(\alpha, \beta))$, $\sign{(B_i(\alpha, \beta))}$ 
  and $\sign{(C_i(\alpha, \beta))}$.  
  Each sign evaluation costs $\sOB(n^{10} + n^{9} \sigma)$,
  using Theorem \ref{th:biv_sign_at} with $n_1 = n$, $n_2 =n^2 $ and $\sigma = n^2 + n \sigma$.
  In the worst case we need $n^2$ of them,
  hence, the cost for all sign evaluations is
  $\sOB( \ell n^{12} +  \ell \,n^{11} \,\sigma )$.
\end{proof}

\subsection{The complexity of topology}

In this section we consider the problem of computing the topology of a real plane
algebraic curve, and improve upon its asymptotic complexity. The reader may refer to
e.g.~\cite{BPR06,VegKah:curve2d:96,mpsttw-aicg-06} for the details of the
algorithm.

We consider the curve in generic position (Section \ref{sec:determ-shear}), 
defined by $F \in \ZZ[x, y]$, such that $\dg( f) = n$ and $\bitsize{ F} = \sigma$.
We compute the critical points of the curve, i.e.\ solve $F = F_y = 0$
in $\sOB( n^{12} + n^{10} \sigma^2)$,
where $F_y$ is the derivative of $F$ w.r.t.~$y$.
Next, we compute the intermediate points on the $x$ axis,
in $\sOB( n^4 + n^3 \sigma)$ (Lemma \ref{lem:very-important}).
For each intermediate point, say $q_j$, we need to compute the number
of branches of the curve that cross the vertical line $x = q_j$.
This is equivalent to computing the number of real solutions 
of the polynomial $F( q_j, y) \in \ZZ[y]$,
which has degree $d$ and bitsize $\OO( n \bitsize{ q_j})$.
For this we use Sturm's theorem and Theorem \ref{pr:SR-fast-evaluation}
and the cost is $\sOB( n^3 \bitsize{ q_j})$.
For all $q_j$'s the cost is $\sOB( n^7 + n^6 \sigma)$.

For each critical point, say $(\alpha, \beta)$
we need to compute the number of branches of the curve that cross 
the vertical line $x = \alpha$, 
and the number of them that are above $y = \beta$.
The first task corresponds to computing the number of real roots
of $F( \alpha, y)$, by application of 
Lemma \ref{lem:count-alpha}, in $\sOB(n^{9} + n^{8} \sigma)$, 
where $n_1 = n$, $n_2 = n^2$, and $\tau = n^2 + n \sigma$.
Since there are $\OO( n^2)$ critical values, the overall cost 
of the step is $\sOB(n^{11} + n^{10} \sigma)$.

Finally, we compute the number of branches that cross the line $x = \alpha$
and are above $y = \beta$ in $\sOB( n^{13} + n^{12} \sigma)$,
by Lemma \ref{lem:count-alpha-beta}.
Since there are $\OO( n^2)$ critical points, the complexity
is $\sOB( n^{15} + n^{14} \sigma)$.
It remains to connect the critical points according to the information that
we have for the branches. The complexity of this step is dominated.
It now follows that the complexity of the algorithm is
$\sOB( n^{15} + n^{14}\sigma + n^{10} \sigma^2)$, or $\sOB( N^{15})$,
which is worse by a factor than \cite{BPR06}.

We improve the complexity of the last step since
\func{m\_rur} computes the RUR representation of the ordinates.
Thus, instead of performing bivariate sign evaluations in order to compute
the number of branches above $y = \beta$, 
we can substitute the RUR representation of $\beta$
and perform univariate sign evaluations.
This corresponds
to computing the sign of $\OO(n^2)$ polynomials
of degree $\OO( n^2)$ and bitsize $\OO( n^4 + n^3 \sigma)$, over all the $\alpha$'s \cite{VegKah:curve2d:96}.
Using Lemma \ref{lem:sign-at-1-all} for each polynomial the cost
is $\sOB( n^{10} + n^{9} \sigma)$, and since there are $\sOB( n^2)$ of them, the total cost
is $\sOB( n^{12} + n^{11} \sigma)$.

\begin{theorem} \label{th:topology}
  We compute the topology of a real plane algebraic curve,
  defined by a polynomial of degree $n$ and bitsize $\sigma$,
  in $\sOB( n^{12} + n^{11} \sigma + n^{10} \sigma^2)$,
  or $\sOB(N^{12})$, where $N = \max\{n, \sigma\}$.
\end{theorem}

Thus the overall complexity of the algorithm %
improves the previously known bound by a factor of $N^2$.  
We assumed generic position, since we can apply a %
shear to achieve this, see Section \ref{sec:grid-solve}.

\input{implementation-new}
\bigskip

\noindent\textbf{Acknowledgements.}
The authors thank the anonymous referees for their comments.
All authors acknowledge partial support by IST Programme of the EU as a Shared-cost
RTD (FET Open) Project under Contract No IST-006413-2 (ACS - Algorithms for Complex Shapes).
The third author is also partially supported by contract ANR-06-BLAN-0074 ''Decotes''.
\bibliographystyle{plainnat}
\bibliography{biv}

\end{document}

%% file: notations-new.tex
\newcommand{\sa}{}

\newcommand{\NN}{\ensuremath{\mathbb N}\xspace}
\newcommand{\ZZ}{\ensuremath{\mathbb Z}\xspace}
\newcommand{\QQ}{\ensuremath{\mathbb Q}\xspace}
\newcommand{\ALG}{\ensuremath{\mathbb{R}_{alg}}\xspace}
\newcommand{\RR}{\ensuremath{\mathbb R}\xspace}

\newcommand{\rem}[1]{\ensuremath{\mathop{\mathsf{rem}}\left( #1 \right)}\xspace}
\newcommand{\quo}[1]{\ensuremath{\mathop{\mathsf{quo}}\left( #1 \right)}\xspace}
\newcommand{\prem}[1]{\ensuremath{\mathop{\mathsf{prem}}\left( #1  \right)}\xspace}

\newcommand{\pgcd}{\ensuremath{\mathsf{gcd}}}

\newcommand{\var}{\ensuremath{\mathtt{VAR}}\xspace}

\newcommand{\dg}{\ensuremath{\mathsf{dg}}\xspace}
\newcommand{\sign}{\ensuremath{\mathop{\mathrm{sign}}}\xspace}

\newcommand{\rat}[1]{\ensuremath{ \mathsf{#1} }\xspace}

\newcommand{\bitsize}[1]{\ensuremath{\mathcal{L}\left( #1 \right)}\xspace}
\newcommand{\Multiply}[1]{\ensuremath{\mathsf{M}\left( #1 \right)}\xspace}

\newcommand{\OO}[1][]{\ensuremath{%
    \ifthenelse{\equal{#1}{}}{\mathcal{O}} {\mathcal{O} \left( #1 \right)\xspace}%
  }}
\newcommand{\sO}[1][]{\ensuremath{%
    \ifthenelse{\equal{#1}{}}{\widetilde{\mathcal{O}}} {\widetilde{\mathcal{O}} \left( #1 \right)\xspace}%
  }}
\newcommand{\OB}[1][]{\ensuremath{%
    \ifthenelse{\equal{#1}{}}{\mathcal{O}_B} {\mathcal{O}_B \left( #1 \right)\xspace}%
  }}
\newcommand{\sOB}[1][]{\ensuremath{%
    \ifthenelse{\equal{#1}{}}{\widetilde{\mathcal{O}}_B} {\widetilde{\mathcal{O}}_B \left( #1 \right)\xspace}%
  }}

\newcommand{\synaps}{\sa{\textsc{synaps}}\xspace}

\newcommand{\newmac}{\sa{\textsc{newmac}}\xspace}

\newcommand{\gbrs}{\sa{\textsc{fgb/rs}}\xspace}

\newcommand{\maple}{\sa{\textsc{maple}}\xspace}

\newcommand{\cpp}{\texttt{\sa{C++}}\xspace}
\newcommand{\cc}{\texttt{\sa{C}}\xspace}

\newcommand{\StHa}{\ensuremath{\mathbf{StHa}}\xspace}

\newcommand{\SR}{\ensuremath{\mathbf{SR}}\xspace}
\newcommand{\sr}{\ensuremath{\mathbf{sr}}\xspace}
\newcommand{\SRQ}{\ensuremath{\mathbf{SRQ}}\xspace}

\newcommand{\res}{\ensuremath{\mathtt{res}}\xspace}

\newcommand{\func}[1]{\ensuremath{\sa{\textsc{#1}}}\xspace}

\newenvironment{myalgorithm}[1][]%
{\begin{algorithm2e}[#1]}%
{\end{algorithm2e}}

\DeclareSymbolFont{EulerScript}{U}{eus}{m}{n}
\SetSymbolFont{EulerScript}{bold}{U}{eus}{b}{n}
\DeclareSymbolFontAlphabet\mathscr{EulerScript}

%% file: implementation-new.tex
\section{Implementation and Experiments}\label{sec:implementation}

This section describes our open source \maple
implementation \footnote{\texttt{www.di.uoa.gr/\textasciitilde erga/soft/SLV\_index.html}} 
and illustrates its capabilities through comparative experiments.
Refer to \cite{det-rr-6116} for its usage and further details.
Our design is object oriented and uses generic programming 
in view of transferring the implementation to \cpp in the future.

We provide algorithms for \emph{signed} polynomial remainder sequences,
real solving of univariate polynomials using Sturm's algorithm,
computations with one and two real algebraic numbers, such as
sign evaluation and comparison and, of course, solving bivariate systems.

\subsection{Our solvers}\label{sec:bivariate-solve}
The performance of all algorithms is averaged over~10 executions on a
\maple~$9.5$ console using a 2GHz AMD64@3K+ processor with $1$GB RAM.
The polynomial systems tested are given in \cite{det-rr-6116}:
systems $R_i, M_i, D_i$ are from \cite{et-casc-2005}, 
the $C_i$ are from \cite{VegNec:topology:02}, and $W_i,i=1,\dots,4$,
follow from $C_i$s after swapping $x, y$.
The latter are of the form $f = \frac{\partial f}{\partial y}= 0$.
For $\gcd$ computations in %
a
(single)
extension field, the package of~\cite{HoeMon:gcd:02} is used.
The optimal algorithms for computing and 
evaluating polynomial remainder sequences have {\em not} yet been implemented.

Our main results are reported in Table \ref{time_table}.
\func{g\_rur} is the solver of choice since it is faster
than \func{grid} and \func{m\_rur} in $17$ out of the $18$ instances.
However, this may not hold when the extension field is of high degree.
\func{g\_rur} yields solutions in $<1$ sec, apart from $C_5$.
For total degree $\leq$ 8, \func{g\_rur} requires $<0.4$ sec.
On average, \func{g\_rur} is $7$-$11$ times faster than \func{grid}, 
and about $38$ times faster than \func{m\_rur}. 
The inefficiency of \func{m\_rur} is due to the fact that it
solves sheared systems which are dense and of increased bitsize;
it also computes multiplicities.
Finally, \func{grid} reaches a stack limit with the default \maple
stack size ($8,192$ KB) when solving $C_5.$ Even when we
multiplied stack size by~10, \func{grid} did not terminate within $20$ min.

Whenever we refer to the speedup we imply the fraction of runtimes.
\func{g\_rur} can be up to $21.58$ times
faster than \func{grid} with an average speedup of around $7.27$ among the input
systems (excluding $C_5$).
With respect to \func{m\_rur}, \func{g\_rur} can be up to $275.74$ times
faster, with an average speedup of $38.01$.

Filtering has been used.
For this, two instances of isolating intervals are stored; 
one for filtering, another for exact computation.
Probably, the most significant filtering technique is interval arithmetic.
When computing the sign of a polynomial evaluated at a 
real algebraic number, the first attempt is via interval arithmetic,
applied along with~\cite{abbott-issac-2006}.
When this fails, and one wants 
to compare algebraic numbers or perform univariate \func{sign\_at},
then the $\gcd$ of two polynomials is computed. 

Filtering helps most with \func{m\_rur}, especially when we compute multiplicities.
With this solver, one more filter is used:
the intervals of candidate $x$-solutions are 
refined by~\cite{abbott-issac-2006} so as to help the interval arithmetic 
filters inside \func{find}.
If the above fails, we switch to exact computation via Sturm sequences,
using the initial endpoints since they have smaller bitsize.
In \func{grid}'s case, filtering provided an average speedup of $1.51$,
where $C_5$ has been excluded.
With \func{g\_rur}, we have on average a speedup of $1.08$.
This is expected since \func{g\_rur} relies heavily on $\gcd$'s and factoring.

\begin{figure}[t]
\begin{center}
\subfloat[Statistics on \func{slv}'s sub-algorithms.]{\label{biv:stat_matrix}
\resizebox{0.492\textwidth}{!}{
\begin{tabular}{|c||l|c|c|c|c|c|}\hline
 & phase of the & 
 \multicolumn{2}{|c|}{interval} & 
\multirow{2}{*}{\small median} & \multirow{2}{*}{mean} & \multicolumn{1}{|c|}{std} \\
 \cline{3-4}
 & algorithm & min     & max     &         &         & \multicolumn{1}{|c|}{dev} \\ \hline \hline
\multirow{4}{*}{\rotatebox{90}{\func{grid}}} 
& projections   & $00.00$ & $00.53$ & $00.04$ & $00.08$ & $00.13$ \\ \cline{2-7}
& univ. solving & $02.05$ & $99.75$ & $07.08$ & $26.77$ & $35.88$ \\ \cline{2-7}
& biv. solving  & $00.19$ & $97.93$ & $96.18$ & $73.03$ & $36.04$ \\ \cline{2-7}
& sorting       & $00.00$ & $01.13$ & $00.06$ & $00.12$ & $00.26$ \\
\hline\hline
\multirow{7}{*}{\rotatebox{90}{\func{mrur}}}
& projection    & $00.00$ & $00.75$ & $00.06$ & $00.14$ & $00.23$ \\ \cline{2-7}
& univ. solving & $00.18$ & $91.37$ & $15.55$ & $17.47$ & $20.79$ \\ \cline{2-7}
& StHa seq.     & $00.08$ & $38.23$ & $01.17$ & $05.80$ & $09.91$ \\ \cline{2-7}
& inter. points & $00.00$ & $03.23$ & $00.09$ & $00.32$ & $00.75$ \\ \cline{2-7}
& filter x-cand & $00.68$ & $72.84$ & $26.68$ & $23.81$ & $21.93$ \\ \cline{2-7}
& compute K     & $00.09$ & $34.37$ & $02.04$ & $07.06$ & $10.21$ \\ \cline{2-7}
& biv. solving  & $01.77$ & $98.32$ & $51.17$ & $45.41$ & $28.71$ \\
\hline\hline
\multirow{6}{*}{\rotatebox{90}{\func{grur}}}
& projections             & $00.02$ & $03.89$ & $00.23$ & $00.48$ & $00.88$ \\ \cline{2-7}
& univ. solving           & $07.99$ & $99.37$ & $39.83$ & $41.68$ & $25.52$ \\ \cline{2-7}
& inter. points           & $00.02$ & $03.81$ & $00.54$ & $01.11$ & $01.28$ \\ \cline{2-7}
& rational biv.           & $00.07$ & $57.07$ & $14.83$ & $15.89$ & $19.81$ \\ \cline{2-7}
& $\mathbb{R}_{alg}$ biv. & $00.00$ & $91.72$ & $65.30$ & $40.53$ & $36.89$ \\ \cline{2-7}
& sorting                 & $00.00$ & $01.50$ & $00.22$ & $00.32$ & $00.43$ \\ \cline{2-7}
\hline
\end{tabular}}}
\hspace{0.02\textwidth}
\subfloat[Performance of our solvers when computing multiplicities.]{\label{time_table_slv_multi}
\resizebox{0.398\textwidth}{!}{
\begin{tabular}{|c||c|c|c||r|r|r|} \hline
\multirow{2}{*}{\textbf{sys}} & 
\multicolumn{2}{|c|}{\multirow{1}{*}{\textbf{deg}}} 
 & $\mathbf{\RR_{alg}}$
 & 
\multicolumn{3}{|c|}{\textbf{Avg Time (msec)}} \\
\cline{2-3} \cline{5-7} %
 & \multicolumn{1}{|c|}{\textbf{$\mathbf{f}$}} &
\multicolumn{1}{|c|}{\textbf{$\mathbf{g}$}} &
\textbf{sols} & 
\multicolumn{1}{|c|}{{\func{grid}}} &
\multicolumn{1}{|c|}{{\func{m\_rur}}} & 
\multicolumn{1}{|c|}{{\func{g\_rur}}} \\ \hline
$R_1$ &  3 &  4 &  2 & $     6$ & $     9$ & $    6$ \\ \hline
$R_2$ &  3 &  1 &  1 & $    66$ & $    21$ & $   36$ \\ \hline
$R_3$ &  3 &  1 &  1 & $     1$ & $     2$ & $    1$ \\ \hline
$M_1$ &  3 &  3 &  4 & $   183$ & $    72$ & $   45$ \\ \hline
$M_2$ &  4 &  2 &  3 & $     4$ & $     5$ & $    4$ \\ \hline
$M_3$ &  6 &  3 &  5 & $ 4,871$ & $   782$ & $  393$ \\ \hline
$M_4$ &  9 & 10 &  2 & $   339$ & $   389$ & $  199$ \\ \hline
$D_1$ &  4 &  5 &  1 & $     6$ & $    12$ & $    6$ \\ \hline
$D_2$ &  2 &  2 &  4 & $   567$ & $   147$ & $  126$ \\ \hline
$C_1$ &  7 &  6 &  6 & $ 1,702$ & $   954$ & $  247$ \\ \hline
$C_2$ &  4 &  3 &  6 & $   400$ & $   234$ & $   99$ \\ \hline
$C_3$ &  8 &  7 & 13 & $   669$ & $ 1,815$ & $  152$  \\ \hline
$C_4$ &  8 &  7 & 17 & $ 7,492$ & $80,650$ & $  474$  \\ \hline
$C_5$ & 16 & 15 & 17 & $ > 20'$ & $60,832$ & $6,367$  \\ \hline
$W_1$ &  7 &  6 &  9 & $ 3,406$ & $ 2,115$ & $  393$ \\ \hline
$W_2$ &  4 &  3 &  5 & $ 1,008$ & $   283$ & $  193$  \\ \hline
$W_3$ &  8 &  7 & 13 & $ 1,769$ & $ 2,333$ & $  230$  \\ \hline
$W_4$ &  8 &  7 & 17 & $ 5,783$ & $77,207$ & $  709$  \\ \hline
\end{tabular}
}
}
\end{center}
\caption{Statistics}\label{statistics}
\end{figure}

Figure \ref{biv:stat_matrix} shows the runtime breakdown corresponding to
the various stages of each algorithm:
\texttt{Projections} shows the time for
computing resultants, \texttt{Univ.Solving} for solving them, and \texttt{Sorting} 
for sorting solutions. In \func{grid}'s and \func{m\_rur}'s case,
\texttt{biv.solving} corresponds to matching. In \func{g\_rur}'s case,
matching is divided between \texttt{rational biv}~and 
\texttt{$\mathbb{R}_{alg}$ biv}; the first refers to when at least one 
of the co-ordinates is rational.
\texttt{Inter.points} refers to computing intermediate points 
between resultant roots along the $y$-axis.
\texttt{StHa seq}~refers to computing the \StHa sequence. 
\texttt{Filter $x$-cand} shows the time for additional filtering. 
\texttt{Compute $K$} reflects the time for sub-algorithm \func{compute-k}.
In a nutshell, \func{grid} spends more than $73\%$ of its time in matching. 
Recall that this percent includes the application of filters and does not take
into account $C_5$.
\func{m\_rur} spends $45$-$50\%$ of its time in matching and $24$-$27\%$ in filtering.
\func{g\_rur} spends $55$-$80\%$ of its time in matching, including
$\gcd$ computations in an extension field.

In order to compute multiplicities, the initial systems were
sheared whenever it was necessary, based on the algorithm
presented in Section \ref{sec:determ-shear}.
Overall results are shown in Figure \ref{time_table_slv_multi}.
\func{grid}'s high complexity starts to become apparent.
Overall, \func{g\_rur} is fastest and
terminates within $\le 1$ sec.
It can be up to $15.81$ times
faster than \func{grid} with an average speedup of around $5.26.$
With respect to \func{m\_rur}, this time \func{g\_rur} can be up to $170.15$ times
faster, with an average speedup of around $18.77$ among all input polynomial systems. 
\func{m\_rur} can be up to $6.23$ times faster than \func{grid}, 
yielding an average speedup of $1.71.$
A detailed table in~\cite{det-rr-6116} 
gives us the runtime decomposition of each algorithm in its major subroutines.
Results are similar to Section \ref{sec:bivariate-solve},
except that \func{g\_rur} spends $68$-$80\%$ of its time in matching, including
$\gcd$'s. In absence of
excessive factoring \func{g\_rur} spends significantly more time in bivariate solving.

\subsection{Other software}\label{sec:biv-all}
\gbrs\ \footnote{\texttt{http://www-spaces.lip6.fr/index.html}} \cite{Rou:rur:99}
performs exact real solving
using Gr{\"o}bner bases and RUR, through its \maple interface;
additional tuning might offer $20$-$30\%$ efficiency increase.
Three \synaps\ \footnote{\texttt{http://www-sop.inria.fr/galaad/logiciels/synaps/}} 
solvers have been tested:
\func{sturm} is a naive implementation of \func{grid}~\cite{et-casc-2005};
\func{subdiv} implements~\cite{MouPav:TR-sbd:05},
using the Bernstein basis and \texttt{double} arithmetic.
It needs an initial box and $[-10,10]\times[-10, 10]$ was used.
\func{newmac}~\cite{MouTre00} is a general purpose solver
based on eigenvectors using \func{lapack}, which computes all complex solutions.

\maple implementations: \func{insulate} implements~\cite{WolSei:topology:05}
for computing the topology of real algebraic curves, and
\func{top} implements \cite{VegNec:topology:02}.
Both packages were kindly provided by their authors.
We tried to modify the packages so as to stop as soon as they compute
the real solutions of the corresponding bivariate system.
It was not easy to modify \func{insulate} and \func{top}
to deal with general systems, so they were not executed on the first data set.
\func{top} has a parameter that sets the initial precision (decimal digits).
There is no easy way for choosing a good value.
Hence, %
recorded its performance on initial values
of $60$ and $500$ digits. 

Experiments are not considered as competition,
but as a crucial step for improving existing software.
It is very difficult to compare different packages, 
since in most cases they are made for different needs.
In addition, accurate timing in \maple is hard, since
it is a general purpose package and a lot of overhead is added to its function calls.
Lastly, the amount of experiments is not very large
in order to draw safe conclusions. 

Overall performance results are shown on Table \ref{time_table}.
In cases where the solvers failed to find the correct number of
real solutions we indicate so with *.
Note that in \newmac 's column an additional step is required to
distinguish the real solutions among the complex ones.
In the sequel we refer only to \func{g\_rur}, since it is our faster implementation.

\begin{table}
\caption{Performance of our solvers and other tested software.
}\label{time_table}
\begin{center}
\resizebox{1.0\textwidth}{!}{
\begin{tabular}{|c|c|c|c|r|r|r|r|r|r|r||r|r|r|} \hline
\multirow{4}{*}{\rotatebox{90}{\textbf{system}}} & 
\multicolumn{2}{|c|}{\multirow{3}{*}{\textbf{deg}}} 
 & \multirow{4}{*}{\rotatebox{90}{\small\textbf{solutions}}}
 & 
\multicolumn{10}{|c|}{\textbf{Average Time (msecs)}} \\
\cline{5-14}
 & \multicolumn{2}{c|}{} 
 & %
  & \multicolumn{7}{|c||}{\textbf{BIVARIATE SOLVING}} &
\multicolumn{3}{|c|}{\textbf{TOPOLOGY}} \\
\cline{5-14}
 & \multicolumn{2}{c|}{} 
 & & %
 \multicolumn{3}{|c|}{\func{slv}} &
\multicolumn{1}{|c|}{\multirow{2}{*}{\small\gbrs}} & 
\multicolumn{3}{|c||}{\small\synaps} &
\multicolumn{1}{|c|}{\multirow{2}{*}{{\small\textsc{insulate}}}} &
\multicolumn{2}{|c|}{{\small\textsc{top}}} \\
\cline{2-3} \cline{5-7} \cline{9-11} \cline{13-14}
 & \multicolumn{1}{|c|}{\textbf{$\mathbf{f}$}} &
\multicolumn{1}{|c|}{\textbf{$\mathbf{g}$}} &
 & \multicolumn{1}{|c|}{\small{\func{grid}}} &
\multicolumn{1}{|c|}{\small{\func{m\_rur}}} & 
\multicolumn{1}{|c|}{\small{\func{g\_rur}}} & &
\multicolumn{1}{|c|}{{\small\func{sturm}}} & 
\multicolumn{1}{|c|}{{\small\func{subdiv}}} & 
\multicolumn{1}{|c||}{{\small \func{newmac}}} & 
 & 
\multicolumn{1}{|c|}{$60$} &
\multicolumn{1}{|c|}{$500$} \\ \hline
$R_1$ &  3 &  4 &  2 & $     5$ & $     9$ & $    5$ & $   26$ &  $     2$  & $       2$  & $    5$ &  $-$       & $-$      & $-$       \\ \hline
$R_2$ &  3 &  1 &  1 & $    66$ & $    21$ & $   36$ & $   24$ &  $     1$  & $       1$  & $    1$ &  $-$       & $-$      & $-$       \\ \hline
$R_3$ &  3 &  1 &  1 & $     1$ & $     2$ & $    1$ & $   22$ &  $     1$  & $       2$  & $    1$ &  $-$       & $-$      & $-$       \\ \hline
$M_1$ &  3 &  3 &  4 & $    87$ & $    72$ & $   10$ & $   25$ &  $     2$  & $       1$  & $    2$ &  $-$       & $-$      & $-$       \\ \hline
$M_2$ &  4 &  2 &  3 & $     4$ & $     5$ & $    4$ & $   24$ &  $     1$  & $     289$* & $    2$ &  $-$       & $-$      & $-$       \\ \hline
$M_3$ &  6 &  3 &  5 & $   803$ & $   782$ & $  110$ & $   30$ &  $   230$  & $   5,058$* & $    7$ &  $-$       & $-$      & $-$       \\ \hline
$M_4$ &  9 & 10 &  2 & $   218$ & $   389$ & $  210$ & $  158$ &  $    90$  & $       3$* & $  447$ &  $-$       & $-$      & $-$       \\ \hline
$D_1$ &  4 &  5 &  1 & $     6$ & $    12$ & $    6$ & $   28$ &  $     2$  & $       5$  & $    8$ &  $-$       & $-$      & $-$       \\ \hline
$D_2$ &  2 &  2 &  4 & $   667$ & $   147$ & $  128$ & $   26$ &  $    21$  & $       1$* & $    2$  &  $-$       & $-$      & $-$       \\ \hline
$C_1$ &  7 &  6 &  6 & $ 1,896$ & $   954$ & $  222$ & $   93$ &  $   479$  & $ 170,265$* & $   39$ &  $    524$ & $   409$ & $  1,367$ \\ \hline
$C_2$ &  4 &  3 &  6 & $   177$ & $   234$ & $   18$ & $   27$ &  $    12$  & $      23$* & $    4$ &  $     28$ & $    36$ & $    115$ \\ \hline
$C_3$ &  8 &  7 & 13 & $   580$ & $ 1,815$ & $   75$ & $   54$ &  $    23$  & $     214$* & $   25$ &  $    327$ & $   693$  & $  2,829$ \\ \hline
$C_4$ &  8 &  7 & 17 & $ 5,903$ & $80,650$ & $  370$ & $  138$ &  $ 3,495$  & $     217$* & $  190$* &  $  1,589$ & $ 1,624$  & $  6,435$ \\ \hline
$C_5$ & 16 & 15 & 17 & $ > 20'$ & $60,832$ & $3,877$ & $4,044$ &  $ > 20'$  & $   6,345$* & $  346$* &  $179,182$ & $91,993$  & $180,917$ \\ \hline
$W_1$ &  7 &  6 &  9 & $ 2,293$ & $ 2,115$ & $  247$ & $   92$ &  $   954$  & $  55,040$* & $   39$ &  $    517$ & $   419$ & $  1,350$ \\ \hline
$W_2$ &  4 &  3 &  5 & $   367$ & $   283$ & $  114$ & $   29$ &  $    20$  & $     224$* & $    3$ &  $     27$ & $    20$  & $     60$ \\ \hline
$W_3$ &  8 &  7 & 13 & $   518$ & $ 2,333$ & $   24$ & $   56$ &  $    32$  & $     285$* & $   25$ &  $    309$ & $   525$  & $  1,588$ \\ \hline
$W_4$ &  8 &  7 & 17 & $ 5,410$ & $77,207$ & $  280$ & $  148$ &  $ 4,086$  & $     280$* & $  207$* &  $  1,579$ & $ 1,458$  & $  4,830$ \\ \hline
\end{tabular}
}
\end{center}
\end{table}

\func{g\_rur} is faster than \gbrs in $8$ out of the $18$ instances, including
$C_5$. The speedup factor ranges from $0.2$ to $22$ with an average of $2.62$.

As for the three solvers from \func{synaps},
\func{g\_rur} is faster than \func{sturm} in $6$ out of the $18$ instances,
but it 
behaves worse usually in systems that
are solved in $<100$ msecs, because \func{sturm} is implemented in \cpp.
As the dimension of the polynomial systems increases, \func{g\_rur} outperforms
\func{sturm} and the latter's lack of fast algorithms for computing resultants
becomes more evident.
Overall, an average speedup of $2.2$ is achieved. %
Compared with \func{subdiv},
\func{g\_rur} is faster in half of the instances
and similarly to the previous case
is slower on systems solved in $< 400$ msecs.
On average, \func{g\_rur} achieves a speedup of $62.92$ which is the result of the
problematic behavior of \func{subdiv} in $C_1$ and $W_1$. If these systems
are omitted, then %
the 
speedup %
is
$8.93$ on average.
\func{newmac} is slower than 
\func{g\_rur} %
in $M_4, D_1$ and $W_3$ and
comparable in $R_1$ and $R_3$.
This time the average speedup %
of our implementation
is $0.53$. %
There are cases where \func{newmac} may not compute some of the real solutions.

Finally, concerning the other \maple software, \func{insulate} is slower than
\func{g\_rur} %
in all systems but $W_2$, %
thus our solver achieves 
an average speedup of $8.85$. 
Compared to \func{top} %
with $60$,  %
resp.~$500$, digits, 
\func{g\_rur} is faster in all
systems but $W_2$, yielding an average speedup of $7.79$,
resp.~$22.64$.
Moreover, as the dimension of the polynomial systems increases, it becomes more efficient.